\documentclass[journal,onecolumn,11pt]{IEEEtran}
\usepackage[utf8]{inputenc}
\usepackage[T1]{fontenc}
\usepackage{url,cite}
\usepackage[cmex10]{amsmath}
\usepackage{etoolbox}
\usepackage{amssymb,mathtools}
\usepackage{amsthm}
\usepackage{graphicx}
\usepackage{multirow}
\usepackage[font = scriptsize, labelfont = bf, labelsep = period, tableposition = top, singlelinecheck = false, justification = centering]{caption}
\usepackage{algorithm2e}
\usepackage{mleftright}       
\mleftright                   

\usepackage{graphicx}         
\usepackage{booktabs}         

\usepackage{arydshln}
\usepackage{bm}
\usepackage{hyperref}
\usepackage{comment}
\includecomment{mycomment}
\newcommand\undermat[2]{%
  \makebox[0pt][l]{$\smash{\underbrace{\phantom{%
    \begin{matrix}#2\end{matrix}}}_{\text{$#1$}}}$}#2}
\newcommand\overmat[2]{%
  \makebox[0pt][l]{$\smash{\overbrace{\phantom{%
    \begin{matrix}#2\end{matrix}}}^{\text{$#1$}}}$}#2}

\usepackage{braket}

\usepackage{enumerate}

\usepackage{tikz}
\usetikzlibrary{tikzmark, shapes, fit,backgrounds}
\usetikzlibrary{matrix,decorations.pathreplacing,angles,quotes}
\usetikzlibrary{patterns, calc, positioning}
\usetikzlibrary{arrows,arrows.meta}
\tikzstyle{block} = [rectangle, draw, 
    minimum width=4em, text centered, rounded corners, minimum height=1em]
\tikzstyle{line} = [draw, -latex]

\newtheorem{definition}{Definition}
\newtheorem{lemma}{Lemma}
\newtheorem{theorem}{Theorem}
\newtheorem{corollary}{Corollary}
\newtheorem{proposition}{Proposition}
\newtheorem{remark}{Remark}
\newtheorem{example}{Example}

\newif\ifcomment
\commenttrue

\newif\ifcommentLater
\commentLaterfalse

\newcommand{\qmarks}[1]{``#1''}
\newcommand{\ie}{\emph{i.e.}, }
\newcommand{\eg}{\emph{e.g.}, }
\definecolor{dartmouthgreen}{rgb}{0.05, 0.5, 0.06}
\newcommand{\commentOut}[1]{}
\newcommand\liteq{\mathrel{\overset{\makebox[0pt]{\mbox{\normalfont\tiny LIT}}}{\equiv}}}

\AtBeginEnvironment{pmatrix}{\setlength{\arraycolsep}{4pt}}
\makeatletter
\renewcommand*\env@matrix[1][*\c@MaxMatrixCols c]{%
  \hskip -\arraycolsep
  \let\@ifnextchar\new@ifnextchar
  \array{#1}}
\makeatother

\DeclareMathOperator{\diagonal}{diag}
\DeclareMathOperator{\trace}{Tr}

\DeclareMathOperator{\rank}{rank}

\DeclareMathOperator{\ttr}{tr}
\DeclareMathOperator{\image}{Im}
\DeclareMathOperator{\LIT}{LIT}
\DeclareMathOperator{\GL}{GL}

\DeclarePairedDelimiter\px{\{}{\}}
\DeclarePairedDelimiter\paren{(}{)}
\DeclarePairedDelimiter\bparen{[}{]}
\DeclarePairedDelimiter\ceil{\lceil}{\rceil}
\DeclarePairedDelimiter\floor{\lfloor}{\rfloor}
\DeclarePairedDelimiterX\symp[2]{\langle}{\rangle_\mathbb{S}}{#1 , #2}
\DeclarePairedDelimiterX\ipp[2]{\langle}{\rangle}{#1 , #2}

\newcommand{\cC}{\mathcal{C}}

\newcommand{\cH}{\mathcal{H}}

\newcommand{\cM}{\mathcal{M}}

\newcommand{\cO}{\mathcal{O}}
\newcommand{\cP}{\mathcal{P}}

\newcommand{\cS}{\mathcal{S}}

\newcommand{\cV}{\mathcal{V}}
\newcommand{\cW}{\mathcal{W}}

\newcommand{\bA}{\mathbf{A}}
\newcommand{\bB}{\mathbf{B}}
\newcommand{\bC}{\mathbf{C}}
\newcommand{\bD}{\mathbf{D}}

\newcommand{\bF}{\mathbf{F}}
\newcommand{\bG}{\mathbf{G}}
\newcommand{\bH}{\mathbf{H}}
\newcommand{\bI}{\mathbf{I}}
\newcommand{\bJ}{\mathbf{J}}

\newcommand{\bM}{\mathbf{M}}
\newcommand{\bN}{\mathbf{N}}

\newcommand{\bP}{\mathbf{P}}
\newcommand{\bQ}{\mathbf{Q}}
\newcommand{\bR}{\mathbf{R}}
\newcommand{\bS}{\mathbf{S}}
\newcommand{\bT}{\mathbf{T}}
\newcommand{\bU}{\mathbf{U}}
\newcommand{\bV}{\mathbf{V}}
\newcommand{\bW}{\mathbf{W}}
\newcommand{\bX}{\mathbf{X}}
\newcommand{\bY}{\mathbf{Y}}
\newcommand{\bZ}{\mathbf{Z}}
\newcommand{\ba}{\mathbf{a}}
\newcommand{\bb}{\mathbf{b}}
\newcommand{\bc}{\mathbf{c}}

\newcommand{\be}{\mathbf{e}}
\newcommand{\bg}{\mathbf{g}}
\newcommand{\bff}{\mathbf{f}}

\newcommand{\bs}{\mathbf{s}}
\newcommand{\bt}{\mathbf{t}}
\newcommand{\bu}{\mathbf{u}}
\newcommand{\bv}{\mathbf{v}}

\newcommand{\bx}{\mathbf{x}}
\newcommand{\by}{\mathbf{y}}

\def\sC{\mathsf{C}}
\def\sH{\mathsf{H}}
\def\sI{\mathsf{I}}
\def\sQ{\mathsf{Q}}

\def\sV{\mathsf{V}}
\def\sX{\mathsf{X}}

\def\sZ{\mathsf{Z}}

\def\a{\alpha}
\def\b{\beta}

\def\bgD{\mathbf{\Delta}}
\def\bgL{\mathbf{\Lambda}}
\def\bgS{\mathbf{\Sigma}}

\def\ppmatrix#1{\begin{pmatrix}#1\end{pmatrix}} 
\def\ppsmatrix#1{\begin{psmallmatrix}#1\end{psmallmatrix}} 

\def\imag#1{\image #1}			
\def\rowspan#1{\langle #1 \rangle_{\mathsf{row}}}
\def\colspan#1{\langle #1 \rangle_{\mathsf{col}}}

\def\bzero{\mathbf{0}}
\def\C{\mathbb{C}}

\def\F{\mathbb{F}}

\def\Fq{\F_{q}}
\def\HWqN{\mathrm{HW}_q^N}
\def\i{\imath}								
\def\N{\mathbb{N}}
\def\R{\mathbb{R}}

\def\symperp{\perp_\mathbb{S}}

\def\zerostate{|\mathbf{0} \rangle}

\usepackage{enumitem}

\begin{document}


\RestyleAlgo{ruled}
\tikzset{meter/.append style={draw, inner sep=10, rectangle, font=\vphantom{A}, minimum width=30, scale=.5, path picture={\draw[black] ([shift={(.1,.3)}]path picture bounding box.south west) to[bend left=50] ([shift={(-.1,.3)}]path picture bounding box.south east);\draw[black,-{Latex[scale=.5]}] ([shift={(0,.1)}]path picture bounding box.south) -- ([shift={(.3,-.1)}]path picture bounding box.north);}}}


\title{$N$-Sum Box: An Abstraction for Linear Computation over Many-to-one Quantum Networks}

\author{
    \IEEEauthorblockN{Matteo Allaix\IEEEauthorrefmark{1}, Yuxiang Lu\IEEEauthorrefmark{2}, Yuhang Yao\IEEEauthorrefmark{2}, Tefjol Pllaha\IEEEauthorrefmark{3}, Camilla Hollanti\IEEEauthorrefmark{1}, Syed Jafar\IEEEauthorrefmark{2}}

	\IEEEauthorblockA{\small \IEEEauthorrefmark{1} Aalto University, Finland. E-mails:
    \{matteo.allaix, camilla.hollanti\}@aalto.fi
	}
	\IEEEauthorblockA{\small \IEEEauthorrefmark{2} University of California, Irvine, CA, USA. E-mails:
    \{yuxiang.lu, yuhangy5, syed\}@uci.edu
	}
	\IEEEauthorblockA{\small \IEEEauthorrefmark{3} University of Nebraska, Lincoln, NE, USA. E-mail:
    tefjol.pllaha@unl.edu
	}
 
    \thanks{
    This work was carried out while M.~Allaix was visiting the research group of S.~Jafar at University of California, Irvine.
    Partial results were submitted to GLOBECOM 2023 \cite{nsumboxarxiv} and to Asilomar 2023 \cite{QCSA23}.
    \emph{The first two authors contributed equally to this work.}
    
    C.~Hollanti and M.~Allaix were supported by the Academy of Finland, under Grant No. 318937. M. Allaix was also supported by a doctoral research grant from the Emil Aaltonen Foundation, Finland. 

    Y.~Lu, Y.~Yao and S.~Jafar were supported by grants NSF CCF-1907053 and CCF-2221379 and ONR N00014-21-1-2386.
    }
}
\IEEEoverridecommandlockouts

\maketitle




\begin{abstract}
Linear computations over quantum many-to-one communication networks offer opportunities for communication cost improvements through schemes that exploit quantum entanglement among transmitters to achieve superdense coding gains, combined with classical techniques such as interference alignment. 
The problem becomes much more broadly accessible if suitable abstractions can be found for the underlying quantum functionality via classical black box models. 
This work formalizes such an abstraction in the form of an \qmarks{$N$-sum box}, a black box generalization of a two-sum protocol of Song \emph{et al.} with recent applications to $N$-server private information retrieval. 
The $N$-sum box has a communication cost of $N$ qudits and classical output of a vector of $N$ $q$-ary digits linearly dependent (via an $N \times 2N$ transfer matrix) on $2N$ classical inputs distributed among $N$ transmitters. 
We characterize which transfer matrices are feasible by our construction, both with and without the possibility of additional locally invertible classical operations at the transmitters and receivers.
Furthermore, we provide a sample application to Cross-Subspace Alignment (CSA) schemes to obtain efficient instances of Quantum Private Information Retrieval (QPIR) and Quantum Secure Distributed Batch Matrix Multiplication (QSDBMM).
We first describe $N$-sum boxes based on maximal stabilizers and we then consider non-maximal-stabilizer-based constructions to obtain an instance of Quantum Symmetric Private Information Retrieval.
\end{abstract}

\section{Introduction}

Distributed computation networks are often limited by their communication costs. 
Improving the efficiency of distributed computation by reducing communication costs is an active area of research. 
Reductions in communication cost may be achieved by coding techniques that are specialized for the type of distributed computation task (e.g., aggregation\cite{nazer2007computation}, MapReduce\cite{dean2008mapreduce}, matrix multiplication\cite{Jia_Jafar_SDMM}) as well as the nature of the communication network (wireless \cite{nazer2007computation}, cable  \cite{ahlswede2000network},
optical fiber \cite{belzner2009network}, quantum networks\cite{hayashi2007quantum}). 
For instance, coding for over-the-air computation reduces the communication cost of linear computation over many-to-one wireless networks, by taking advantage of the natural superposition property of the wireless medium \cite{nazer2007computation}.

Investigating the properties and the applications of quantum protocols is crucial for the development of a quantum internet \cite{cacciapuotiteleportation2020,cacciapuotiquantuinternet2020,caleffi2022distributed}, which operates on the principles of quantum mechanics and differs fundamentally from the classical internet used in our daily lives.
Our focus in this work is on  \emph{linear} computations (possibly with privacy and security constraints) over \emph{quantum} many-to-one communication networks. 
The potential for reduced communication costs in this setting comes from quantum entanglement among the transmitters, which creates opportunities for superdense coding gains\cite{bennett1992communication, werner2001all, liu2002general, gorbachev2002teleportation} as well as classical techniques such as interference alignment. 
However, unlike wireless networks for which there exists an abundance of simplified channel models and  abstractions to facilitate analysis from coding, information-theoretic and signal-processing perspectives \cite{el2011network, tse2005fundamentals, goldsmith2005wireless}, similarly convenient abstractions of quantum communication networks are not readily available, which limits the study of quantum communication networks largely to quantum experts. 
Our work is motivated by the observation that a convenient abstraction for linear computation over quantum many-to-one networks is indeed available, although somewhat implicitly, in the works of Song \emph{et al.}, in the form of a quantum two-sum protocol~\cite{song2019capacity, song2019capacitycollusion}, and its subsequent generalizations, as applied to QPIR \cite{QTPIR,allaix2020quantum, allaix2021quantum,  QMDSTPIR, QQPIR21}.
A similar type of linear computation with a classical-quantum multiple access channel has been discussed in \cite{hayashi2021computation,sohail2022computing,sohail2022unified}, where the classical-quantum multiple access channel might have noise, and further applied to the problem of quantum secret sharing and QPIR \cite{hayashi2022unified}.
\subsection{Two-Sum Protocol}
The two-sum protocol\cite{song2019capacitycollusion} is shown in Figure~\ref{fig:twosumbox}, both as a quantum circuit and as a black box. 
\begin{figure}[t]\setlength{\hfuzz}{1.1\columnwidth}
\begin{minipage}{\textwidth}
\centering
\begin{tikzpicture}
\node (Tx1) at (1.5,0.15) [rectangle, fill=red!10, minimum width=2.5cm, minimum height=1cm]{};
\node [above =0cm of Tx1] {\tiny \bf Tx1};
\node (Tx2) at (1.5,-1.15) [rectangle, fill=blue!10, minimum width=2.5cm, minimum height=1cm]{};
\node [below =0cm of Tx2] {\tiny \bf Tx2};
\node (Rx) at (7,-0.5) [rectangle, fill=orange!10, minimum width=4.5cm, minimum height=1.75cm]{};
\node [below =0cm of Rx] {\footnotesize Receiver};
\node (phi) at (-0.2,-0.5){\footnotesize $\ket{\b_{00}}$};
\node (Z1) at (1,0) [draw, rectangle, inner sep =0.1cm] {$\sZ$} ;
\node (Z2) at (1,-1) [draw, rectangle, inner sep =0.1cm] {$\sZ$};
\node (X1) at (2,0) [draw, rectangle, inner sep =0.1cm] {$\sX$};
\node (X2) at (2,-1) [draw, rectangle, inner sep =0.1cm] {$\sX$};
\node (a2) at (Z1) [above=0.2cm] {\footnotesize $x_3$};
\node (a1) at (X1) [above=0.2cm] {\footnotesize $x_1$};
\node (b2) at (Z2) [below=0.2cm] {\footnotesize $x_4$};
\node (b1) at (X2) [below=0.2cm] {\footnotesize $x_2$};

\begin{scope}[shift={(-0.75,0)}]
\node (x2) at (9.1,0) [inner sep = 0.1cm]{\footnotesize $x_1+x_2$};
\node (y2) at (9.1,-1) [inner sep =0.1cm] {\footnotesize $x_3+x_4$};
\node (H2) at (6.7,0) [draw, rectangle, inner sep =0.1cm] {$\sH$};
\node (mx2) at (7.8,0) [meter] {};
\node (my2) at (7.8,-1) [meter] {};
\node (b2) at (6,0) [draw, circle, inner sep=0.05cm, fill=black]{} ;
\coordinate (xend2) at (5.4,0);
\coordinate (yend2) at (5.4,-1);
\node (op2) at (6,-1) [inner sep=0] {$\oplus$};
\draw [color=black, thick] (b2)--(op2);
\draw [color=black, thick] (xend2)--(H2)--(mx2)--(x2);
\draw [color=black, thick] (yend2)--(op2)--(my2);
\draw [double] (mx2)--(x2);
\draw [double] (my2)--(y2);
\end{scope}
\draw [color=black, thick] (0,0)--(Z1)--(X1)--(3,0);
\draw [color=black, thick] (0,-1)--(Z2)--(X2)--(3,-1);
\draw [color=black, dashed] (3,0)--(xend2) node [midway, align=center] {\footnotesize $1$ qubit\\ \footnotesize to Rx};
\draw [color=black, dashed] (3,-1)--(yend2) node [midway, align=center] {\footnotesize $1$ qubit\\ \footnotesize to Rx};

\node (box) at (7.7,-0.5) [rectangle, draw, help lines, minimum width=17.1cm, minimum height=3.5cm]{};
\draw [] ($(box.north)+(1.75,0)$) -- ($(box.south)+(1.75,0)$);

\begin{scope}[shift={(10.6,0)}]
\node (Tx1) at (0,0.1) [ help lines, text =black, rectangle, minimum height=0.9cm, minimum width=0.75cm, fill=red!10] {};
\node [above =0cm of Tx1] {\tiny \bf Tx1};
\node [left=2pt of {$(Tx1.east)+(0,0.25)$}] {\footnotesize $x_1$};
\node [left=2pt of {$(Tx1.east)-(0,0.25)$}] {\footnotesize $x_3$};

\node (Tx2) at (0,-1.1) [ help lines, rectangle, minimum height=0.9cm, minimum width=0.75cm, text =black,  fill=blue!10]{};
\node [below =0cm of Tx2] {\tiny \bf Tx2};
\node [left=2pt of {$(Tx2.east)+(0,0.25)$}] {\footnotesize $x_2$};
\node [left=2pt of {$(Tx2.east)-(0,0.25)$}] {\footnotesize $x_4$};

\node (R) at (4.3,-0.5) [ help lines, rectangle, text=black, fill=orange!10, minimum width=2.3cm, minimum height=1.4cm]{};
\node [below =0cm of R] {\footnotesize Receiver};
\node [right=0.1cm of {$(R.west)+(0,0.3)$}] {\footnotesize $y_1=x_1+x_2$};
\node [right=0.1cm of {$(R.west)-(0,0.3)$}] {\footnotesize $y_2=x_3+x_4$};

\node (G) at (1.75,-0.5) [draw,  very thick, rectangle, minimum width=1.6cm, minimum height=2cm, fill=black!10, align=center] { \footnotesize $2$-Sum\\[-0.1cm] \footnotesize Box\\[-0.3cm]~\\  \footnotesize Com. Cost \\[-0.1cm] \footnotesize 2 qubits};

\draw [thick, ->] let \p{Geast}=($(G.east)+(0,0.3)$), \p{Rwest}=(R.west) in (\x{Geast},\y{Geast})--(\x{Rwest},\y{Geast}) node [above, midway] {};
\draw [thick, ->] let \p{Geast}=($(G.east)-(0,0.3)$), \p{Rwest}=(R.west) in (\x{Geast},\y{Geast})--(\x{Rwest},\y{Geast}) node [below, midway] {};

\draw [thick, ->] let \p{Tx1east}=($(Tx1.east)+(0,0.25)$), \p{Gwest}=(G.west) in (\x{Tx1east},\y{Tx1east})--(\x{Gwest},\y{Tx1east}) node [above, midway] {}; 
\draw [thick,->] let \p{Tx1east}=($(Tx1.east)-(0,0.25)$), \p{Gwest}=(G.west) in (\x{Tx1east},\y{Tx1east})--(\x{Gwest},\y{Tx1east}) node [above, midway] {}; 
\draw [thick, ->] let \p{Tx2east}=($(Tx2.east)+(0,0.25)$), \p{Gwest}=(G.west) in (\x{Tx2east},\y{Tx2east})--(\x{Gwest},\y{Tx2east}) node [below, midway] {};
\draw [thick, ->] let \p{Tx2east}=($(Tx2.east)-(0,0.25)$), \p{Gwest}=(G.west) in (\x{Tx2east},\y{Tx2east})--(\x{Gwest},\y{Tx2east}) node [below, midway]{};
\end{scope}
\end{tikzpicture}
\caption{Quantum circuit and black-box representation for two-sum transmission protocol with $\ket{\b_{00}} = \frac{1}{\sqrt{2}} (\ket{00} + \ket{11})$.}
\label{fig:twosumbox}
\end{minipage}
\end{figure}
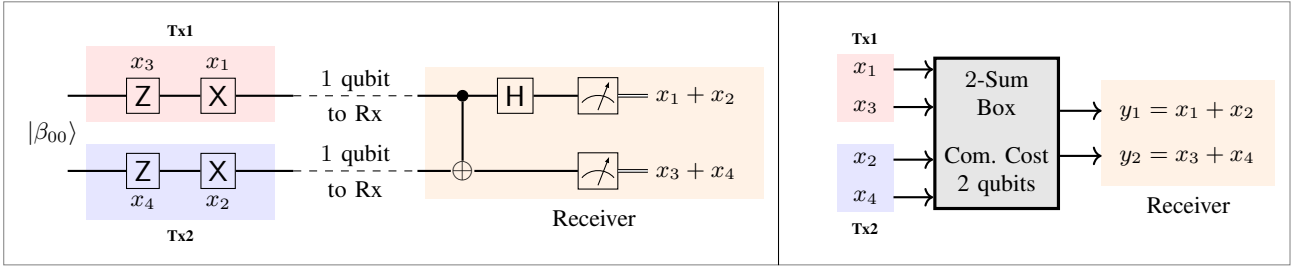
In the quantum circuit, we see two transmitters (Tx1 and Tx2), each in possession of one qubit of an entangled pair. 
The entangled state in this case is the Bell state $\ket{\beta_{00}}$. 
As classical $2$-bit inputs become available to the two transmitters ($(x_1,x_3)$ to Tx1, $(x_2,x_4)$ to Tx2), they perform conditional quantum operations (${\sf X}, {\sf Z}$ gates) on their respective qubits and then send them to the receiver (Rx), for a total communication cost of $2$ qubits. 
The receiver performs a Bell measurement and obtains $(y_1,y_2)=(x_1+x_2,x_3+x_4)$. 
The two-sum protocol can be abstracted into a black box, also shown in Figure~\ref{fig:twosumbox}, with inputs $(x_1,x_3), (x_2,x_4)$ controlled by Tx1 and Tx2, respectively, and output $\by=\bM\bx$, where $\bM = \ppsmatrix{1 & 1 & 0 & 0 \\ 0 & 0 & 1 & 1}$ is  the \emph{transfer matrix} of this $2$-sum box and $\bx^\top = \ppmatrix{x_1,x_2,x_3,x_4}$. 
The black-box representation hides the details of the quantum circuit and specifies only the functionality (transfer matrix $\bM$) and the communication cost ($2$ qubits), which makes it possible for non-quantum experts to design low-communication-cost coding schemes for quantum communication networks using this black box, \eg to take advantage of super-dense coding. 
Note that in the two-sum protocol, if Tx1 uses only zeros for its data $(x_1=x_3=0)$, then the protocol allows Tx2 to send both its classical input bits $(x_2,x_4)$ to the receiver, even though it sends only one qubit, provided Tx1 sends its qubit to the receiver as well. This is an example of superdense coding, made possible by the entanglement between the two qubits. 
Indeed, without the entanglement, qubits are worth no more than classical bits by the Holevo bound~\cite{holevo1973bounds}. 


\subsection{Our Contribution}

The main contribution of this work is to formalize a generalization of the $2$-sum box, namely an $N$-sum box. 
It is worth pointing out immediately that the technical foundations of the $N$-sum box are not new, indeed the construction draws upon the well-understood stabilizer formalism in quantum coding theory\cite{Gottesman97, AK01, Ketkar06, CRSS98}, and most of the technical details of  the generalization from $2$-sum to $N$-sum are also contained in the works of Song and Hayashi on Quantum Private Information Retrieval \cite{QTPIR}. 
Nevertheless, the crystallization of the black-box abstraction, as demonstrated in this study, holds significant promise for researchers in the classical information and coding theory domains. 
These researchers, though less acquainted with stabilizer codes and quantum coding theory, can still make valuable contributions to comprehending the fundamental boundaries of transmitter-side entanglement-assisted distributed classical computation over quantum multiple access (QMAC) networks. 
This is achieved through the utilization of the aforementioned classical abstraction, which effectively conceals the intricate details of the underlying quantum circuitry.
For example, wireless researchers with little background in quantum codes may recognize the $N$-sum box as the familiar MIMO MAC setting illustrated via an example in Figure~\ref{fig:MIMOMAC}. 
The main distinctions from the multiple antenna wireless setting are 1) that the channel is deterministic (noise-free), defined over a finite field ($\Fq$) rather than complex numbers, and 2) that instead of being generated randomly by nature, the channel matrix can be freely designed as long as it is strongly self orthogonal (SSO) (cf. Definition~\ref{def:ssomatrix}). 
This is because it is shown in this work that feasible $N$-sum box transfer functions are precisely those  matrices $\bM\in\F_q^{N\times 2N}$ that are either strongly self-orthogonal themselves, or can be made strongly self-orthogonal by local invertible transformations (cf. Definition~\ref{def:lit}) at various transmitters and/or the receiver. 
Thus, from a wireless perspective, the problem of coding for the QMAC becomes conceptually equivalent to that of designing a coding scheme as well as the channel matrix for a MIMO MAC subject to given structural constraints imposed by the $N$-sum box abstraction (SSO), such that the resulting MIMO MAC is able to efficiently achieve the desired linear computation `over-the-air' (actually, through quantum entanglement). 
The efficiency gained by `over-the-air' computation in this (constrained: SSO) MIMO MAC translates into  superdense coding gain over the QMAC. 

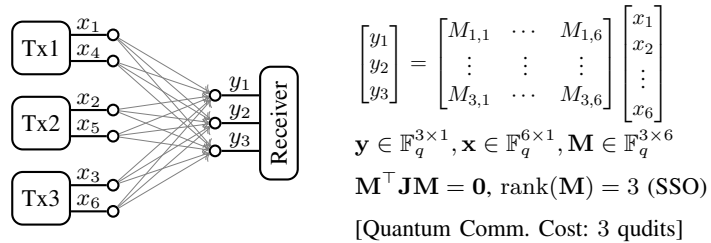
\begin{figure}[h]
\begin{center}
\begin{tikzpicture}
\node (Tx1) [draw, rectangle, rounded corners, thick, minimum width=0.52cm, minimum height=0.7cm] {\footnotesize Tx1};
\node (x1) [draw, thick, circle, inner sep = 0.05cm, above right=-0.25cm and 0.5cm of Tx1] {};
\node (x4) [draw, thick, circle, inner sep = 0.05cm, below right=-0.25cm and 0.5cm of Tx1] {};
\draw [thick] (x1.west)--(Tx1.east|-x1.west) node [midway, above=-0.1cm]{\footnotesize $x_1$};
\draw [thick] (x4.west)--(Tx1.east|-x4.west) node [midway, above=-0.1cm]{\footnotesize $x_{4}$};
\begin{scope}[yshift=-1cm]
\node (Tx2) [draw, rectangle, rounded corners, thick, minimum width=0.52cm, minimum height=0.7cm] {\footnotesize Tx2};
\node (x2) [draw, thick, circle, inner sep = 0.05cm, above right=-0.25cm and 0.5cm of Tx2] {};
\node (x5) [draw, thick, circle, inner sep = 0.05cm, below right=-0.25cm and 0.5cm of Tx2] {};
\draw [thick] (x2.west)--(Tx2.east|-x2.west) node [midway, above=-0.1cm]{\footnotesize $x_2$};
\draw [thick] (x5.west)--(Tx2.east|-x5.west) node [midway, above=-0.1cm]{\footnotesize $x_{5}$};
\end{scope}
\begin{scope}[yshift=-2cm]
\node (Tx3) [draw, rectangle, rounded corners, thick, minimum width=0.52cm, minimum height=0.7cm] {\footnotesize Tx3};
\node (x3) [draw, thick, circle, inner sep = 0.05cm, above right=-0.25cm and 0.5cm of Tx3] {};
\node (x6) [draw, thick, circle, inner sep = 0.05cm, below right=-0.25cm and 0.5cm of Tx3] {};
\draw [thick] (x3.west)--(Tx3.east|-x3.west) node [midway, above=-0.1cm]{\footnotesize $x_3$};
\draw [thick] (x6.west)--(Tx3.east|-x6.west) node [midway, above=-0.1cm]{\footnotesize $x_{6}$};
\end{scope}

\node (Rx) [draw, rectangle, rounded corners, thick, minimum width=0.5cm, minimum height=1.5cm, right=2.5cm of Tx2] {\footnotesize \rotatebox{90}{Receiver}};
\node (y2) [draw, thick, circle, inner sep = 0.05cm,   left= 0.5cm of Rx] {};
\node (y1) [draw, thick, circle, inner sep = 0.05cm,  above= 0.2cm of y2] {};
\node (y3) [draw, thick, circle, inner sep = 0.05cm,  below= 0.2cm of y2] {};

\draw [thick] (y1.east)--(Rx.west|-y1.east) node [midway, above=-0.1cm]{\footnotesize $y_1$};
\draw [thick] (y2.east)--(Rx.west|-y2.east) node [midway, above=-0.1cm]{\footnotesize $y_2$};
\draw [thick] (y3.east)--(Rx.west|-y3.east) node [midway, above=-0.1cm]{\footnotesize $y_3$};

\foreach \i in {1,2,...,6},
	\foreach \j in {1,2,3},
		\draw [->, help lines] (x\i)--(y\j);
		
\node [right=0.61cm of Rx, align=left]{\scalebox{0.75}{$\begin{bmatrix}y_1\\y_2\\y_3\end{bmatrix}=\begin{bmatrix}M_{1,1}&\cdots&M_{1,6}\\ \vdots&\vdots&\vdots\\ M_{3,1}&\cdots&M_{3,6}\end{bmatrix}\begin{bmatrix}x_1\\x_2\\\vdots\\x_6\end{bmatrix}$}\\[0.1cm] \footnotesize ${\bf y}\in\mathbb{F}_q^{3\times 1}, {\bf x}\in\mathbb{F}_q^{6\times 1}, \bM\in\mathbb{F}_q^{3\times 6}$\\[0.1cm]\footnotesize $\bM^\top{\bf J}\bM=\bzero$, $\rank(\bM)=3$ (SSO)\\[0.1cm] \footnotesize {\color{black}[Quantum Comm. Cost: $3$ qudits]}};
\end{tikzpicture}
\end{center}
\caption{The $N$-sum box is illustrated as a MIMO MAC  for $N=3$.}\label{fig:MIMOMAC}
\end{figure}

The $N$-sum box is intended to be useful primarily as a tool for exploring the information-theoretic capacity of $\Fq$-linear classical computations over an ideal QMAC, with the potential to shed new light into the fundamental limitations of superdense coding and quantum entanglement. 
As with other tools that information theorists have at their disposal, it is difficult to predict in advance if the $N$-sum box abstraction will turn out to be sufficient to construct capacity achieving schemes. 
Indeed  the linear computation capacity of a MAC is a challenging problem even in the classical setting, especially for \emph{vector} linear computations. 
Nevertheless, we are cautiously optimistic that the stabilizer-based construction exhausts the scope of the $N$-sum box functionality for $\Fq$-linear computations. 
The optimistic outlook is supported by prior works on capacity of QPIR \cite{QTPIR, allaix2021quantum,QMDSTPIR} where $N$-sum boxes have been implicitly employed for capacity-achieving schemes, as well as a recent follow-up work that utilizes the $N$-sum-box abstraction from this work, to find the capacity of sum-computation over the QMAC \cite{SQMAC23}. 

Last but not the least, even when an exact capacity characterization is beyond reach, a fruitful strategy is to utilize the $N$-sum box to design the best possible schemes allowed by the abstraction. 
In general the constraints of the abstraction may lead to entirely new schemes. 
However, certain applications of interest, of which QPIR is a prime example, have classical solutions with specialized structures that naturally resonate with the SSO constraint, the defining feature of $N$-sum boxes. 
For such applications, it can be particularly insightful to find ways to efficiently \emph{quantumize} the classical solutions, leading not only to good quantum coding schemes but also a better understanding of the role of the SSO structure for linear computations. 
Notably, such a quantumization was introduced in~\cite{allaix2020quantum} by blending the star-product scheme \cite{freij2017private} with the two-sum protocol. 
In this work, to further illustrate this aspect, we provide another instance by quantumizing classical cross-subspace alignment (CSA) codes into QCSA codes. CSA codes have been used in a variety of schemes ranging from XSTPIR~\cite{Jia_Sun_Jafar_XSTPIR} and MDS coded XSTPIR~\cite{Jia_Jafar_MDSXSTPIR} to secure distributed batch matrix multiplication (SDBMM)~\cite{Jia_Jafar_SDMM,Chen_Jia_Wang_Jafar_NGCSA}. 
Therefore, QCSA codes naturally open the door for the general quantum MDS-coded XSTPIR (MDS-coded QXSTPIR) setting as well as quantum SDBMM (QSDBMM). 


\subsection{Organization of the Paper}

In Section~\ref{sec:stab_formalism} we describe the stabilizer formalism over a finite field, which provides the quantum building block for stabilizer-based $N$-sum-box constructions.
In Section~\ref{sec:n_sum_box} we formally define an $N$-sum box.
In Section~\ref{sec:sbnsumbox} we study which constructions based on a maximal stabilizer are feasible both allowing and disallowing local invertible transformations, \ie invertible transformations applied to the inputs by the transmitters or invertible transformations applied to the output by the receiver.
In Section~\ref{sec:QCSA} we provide an application to CSA schemes to obtain a QCSA scheme that enables instances of MDS-coded QXSTPIR and QSDBMM.
Finally in Section~\ref{sec:non_max_stab_based_constructions} we consider constructions based on non-maximal stabilizers that enable instances of symmetric MDS-coded QXSTPIR without the need for shared randomness among the servers.


\subsection{Notations}
We denote by $[N]$ the set $\px*{1,\ldots,N},\ n \in \N$, and by $\Fq$ the finite field with $q$ elements. 
We use bold lower-case letters and bold upper-case letters to denote vectors and matrices, respectively. 
Given a matrix $\bA$, $\rowspan{\bA}$ and $\colspan{\bA}$ denote the spaces spanned by the rows and columns of $\bA$, respectively, while $\bA^\top$ and $\bA^\dagger$ represent its transpose and its conjugate transpose, respectively.


\section{Stabilizer formalism over finite fields}
\label{sec:stab_formalism}

The stabilizer formalism~\cite{Gottesman97} is a compact framework for quantum computation that provides a useful bridge to classical computation. 
Recently, this framework has been leveraged to boost several classical protocols.
We describe the stabilizer formalism over a finite field, for the details of which we refer the reader to~\cite{AK01,Ketkar06}. 
Throughout, we will use the same notation as in~\cite{QMDSTPIR}.


Let $q = p^r$ with a prime number $p$ and a positive integer $r$.
Let $\cH$ be a $q$-dimensional Hilbert space spanned by orthonormal states $\px*{ |j\rangle :  j\in \Fq }$. 
For $x \in \Fq$, we define  $\bT_x$ on $\F_p^{r}$ as the $\F_p$-linear map $y \mapsto xy \in\Fq$, $y\in\Fq$, by identifying the finite field $\Fq$ with the vector space $\F_p^{r}$.
Let $\ttr x \coloneqq \trace \bT_x \in \F_p$ for $x\in\Fq$.
Let $\omega \coloneqq \exp({2\pi i/p})$.
For $a,b\in\Fq$, we define unitary matrices $\sX(a) \coloneqq \sum_{j\in\Fq} |j+a\rangle \langle j |$ and $\sZ(b) \coloneqq \sum_{j\in\Fq} \omega^{\ttr bj} |j\rangle \langle j |$ on $\cH$. 
For $\bs = (s_1,\ldots, s_{2N}) \in \Fq^{2N}$, we define a unitary matrix $\tilde{\bW}(\bs) \coloneqq \sX(s_1) \sZ(s_{N+1}) \otimes \cdots \otimes \sX(s_{N}) \sZ(s_{2N})$ on $\cH^{\otimes N}$ called \emph{Weyl operator}. 

For $\bx=(x_1,\ldots,x_N),\ \by=(y_1,\ldots,y_N) \in \Fq^{N}$, we define the tracial bilinear form $\langle \bx, \by \rangle \coloneqq \ttr \sum_{i=1}^{N} x_iy_i\in\F_p$ and the trace-symplectic bilinear form $\symp{\bx}{\by} \coloneqq \langle \bx, \bJ\by \rangle$, where $\bJ$ is the $2N \times 2N$ matrix 
\[
\bJ = \ppmatrix{\bzero & -\bI \\ \bI & \bzero}.
\]
The dual of a subspace $\cV$ of $\Fq^{2N}$ with respect to this form is $\cV^{\symperp} \coloneqq \px*{ \bs\in \Fq^{2N} :  \symp{\bv}{\bs} = 0 \text{ for any } \bv\in \cV }$.

A matrix $\bF\in \Fq^{2N\times 2N}$ is called {\em symplectic} if $\bF^\top \bJ \bF = \bJ$. 
Symplectic matrices are precisely those matrices that preserve $\symp{\cdot}{\cdot}$, and its columns form a symplectic basis for $\Fq^{2N}$.
If we write $\bF = \ppsmatrix{\bA & \bC \\ \bB & \bD}$, then $\bF$ is symplectic if and only if $\bB^\top\bA, \bD^\top\bC$ are symmetric and $\bA^\top\bD -\bB^\top\bC = \bI$.
Thus,
\begin{equation}\label{eq:SympInv}
\bF^{-1} = \bJ^\top\bF^\top\bJ = \begin{pmatrix}\bD^\top&-\bC^\top\\-\bB^\top&\bA^\top\end{pmatrix}.
\end{equation}

\begin{remark}
    \label{rem:sso_to_symplectic}
    If $\bF \in \Fq^{2N \times 2N}$ is a symplectic matrix, then it is easy to see that the matrix $\bF' = \bF \ppsmatrix{\bI_\kappa \\ \bzero} \in \Fq^{2N \times \kappa}$, \ie the matrix containing the first $\kappa \in [N]$ columns of $\bF$, satisfies the relation $\paren*{\bF'}^\top \bJ \bF' = \bzero$. 
    Conversely, a matrix satisfying such relation can be completed to a symplectic matrix (\eg using the Gram--Schmidt Algorithm).
\end{remark}

Symplectic orthogonality in the vector space $\Fq^{2N}$ is equivalent to commutativity in the \emph{Heisenberg-Weyl group} $\HWqN \coloneqq \px*{c \tilde{\bW}(\bs) :  \bs \in \Fq^{2N},\  c \in \C \setminus \px*{0} }$. 
In fact, there is a surjective homomorphism $c \tilde{\bW}(\bs)\in \HWqN \mapsto \bs \in \Fq^{2N}$ with kernel $\px*{ c \bI_{q^N} : c \in \C \setminus \px*{0} }$, so two matrices $c_1\tilde{\bW}(\bs_1), c_2\tilde{\bW}(\bs_2)$ commute if and only if $\symp{\bs_1}{\bs_2} = 0$.


A commutative subgroup of $\mathrm{HW}_q^N$ not containing $c \bI_{q^N}$ for any $c\neq 1$ is called a {\em stabilizer group}. Such groups are precisely those groups for which the aforementioned homomorphism is actually an isomorphism. Thus, a stabilizer group defines a {\em self-orthogonal subspace}, that is $\cV\subseteq \cV^{\symperp}$, in $\Fq^{2N}$. 
Conversely, given a self-orthogonal subspace $\cV$ of $\Fq^{2N}$, there exist complex numbers $c_\bv$ so that 
\begin{equation}
    \cS(\cV)  \coloneqq \px*{ \bW(\bv)  \coloneqq c_{\bv} \tilde{\bW}(\bv) :  \bv \in \cV } \subseteq \mathrm{HW}_q^N
    \label{eq:stab}
\end{equation}
forms a stabilizer group. Thus, there is a one-to-one correspondence between stabilizer groups in $\HWqN$ and self-orthogonal subspaces in $\Fq^{2N}$.

Throughout this paper, we focus primarily on {\em maximal} stabilizers, since they exhaust the scope of all possible stabilizer-based $N$-sum boxes (cf. Remark~\ref{rem:max_to_non_max}).
Maximal stabilizers are in one-to-one correspondence with \emph{strongly self-orthogonal (SSO) subspaces}, \ie $\cV = \cV^{\symperp}$, so $\dim(\cV) = N$. 
In Section~\ref{sec:non_max_stab_based_constructions} we consider non-maximal stabilizer-based constructions for the cases when we want to discard $N-\kappa$ of the $N$ outputs of an $N$-sum box.

While $\cV$ defines the stabilizer $\cS(\cV)$, the quotient space $\Fq^{2N}/\cV^{\symperp}$ defines orthogonal projectors 
\begin{equation}
    \label{eq:coset_measurement}
    \cP^\cV \coloneqq \px*{\bP^{\cV}_{\overline{\bs}} : \overline{\bs}\in\Fq^{2N}/\cV^{\symperp}}
\end{equation}
which we use as a projective-value measurement (PVM). 
We will denote $|\overline{\bs}\rangle$ the state which $\bP^{\cV}_{\overline{\bs}}$ projects onto.
Throughout this paper we will use the results shown in the following proposition.

\begin{proposition}{\cite[Proposition~2.2]{QMDSTPIR}}
    \label{prop:stab}
    Let $\cV$ be a $D$-dimensional self-orthogonal subspace of $\Fq^{2N}$ and $\cS(\cV)$ be a stabilizer defined from $\cV$.
    Then, we obtain the following statements.
    
    \begin{enumerate}[label=(\alph*)]
        \item For any $\bv\in\cV$, the operation $\bW(\bv) \in \cS(\cV)$ is simultaneously and uniquely decomposed as 
        \begin{equation}
            \bW(\bv) = \sum_{\overline{\bs}\in\F_q^{2N} / \cV^{\symperp} } \omega^{ \symp{\bv}{\bs} } \bP^{\cV}_{\overline{\bs}}
        	\label{eq:weyldecomp}
        \end{equation}
        with orthogonal projections $\px*{\bP^{\cV}_{\overline{\bs}}}$ such that
        \begin{align*}
        	\bP^{\cV}_{\overline{\bs}} \bP^{\cV}_{\overline{\bt}} &= \bzero \mbox{ for any } \overline{\bs}\neq \overline{\bt}, \\
        	\sum_{\overline{\bs}\in\F_q^{2N} / \cV^{\symperp} } \bP^{\cV}_{\overline{\bs}} &= \bI_{q^N}.
        \end{align*} 
        
        \item Let $\cH^{\cV}_{\overline{\bs}} \coloneqq \imag{\bP^{\cV}_{\overline{\bs}}}$.
        We have $\dim \cH^{\cV}_{\overline{\bs}} = q^{N-D}$ for any $\overline{\bs}\in\F_q^{2N} / \cV^{\symperp}$ and the quantum system $\cH^{\otimes N}$ is decomposed as
        \begin{equation}
            \cH^{\otimes N} = \bigotimes_{\overline{\bs}\in \F_q^{2N} / \cV^{\symperp}} \cH^{\cV}_{\overline{\bs}} = \cW \otimes \C^{q^{N-D}},
            \label{eq:Hdecomposition}
        \end{equation}
        where the system $\cW$ is the $q^D$-dimensional Hilbert space spanned by $\px*{ |\overline{\bs}\rangle  : \overline{\bs}\in \F_q^{2N} / \cV^{\symperp} }$ with the property 
        $\cH^{\cV}_{\overline{\bs}} = |\overline{\bs} \rangle \otimes  \C^{q^{N-D}} \coloneqq \px*{ |\overline{\bs} \rangle \otimes | \psi \rangle : |\psi\rangle \in\C^{q^{N-D}} }$.

        \item For any $\bs, \bt\in\F_q^{2N}$, we have
        \begin{align*}
        	\bW(\bt) |\overline{\bs} \rangle \otimes  \C^{q^{N-D}} 
        	&= |\overline{\bs+\bt} \rangle \otimes  \C^{q^{N-D}},
        	\\
        	\bW(\bt) \paren*{|\overline{\bs} \rangle\langle \overline{\bs} | \otimes  \bI_{q^{N-D}}  } \bW(\bt)^{\dagger}
        	&= |\overline{\bs+\bt} \rangle \langle \overline{\bs+\bt} | \otimes  \bI_{q^{N-D}}.
        \end{align*}
    \end{enumerate}
\end{proposition}

Measuring with $\cP^\cV$ as in Equation~\eqref{eq:coset_measurement} would yield a coset.
Proposition~\ref{prop:measurementoutcome} aims to clarify the notation of Proposition~\ref{prop:stab} by giving a unique representative of the outputted equivalence class. 
First, we need the following lemma to prove that matrices satisfying the conditions of the proposition exist.

\begin{lemma}
    \label{lem:Gperp}
    Let $\bG \in \Fq^{2N \times \kappa}$ be such that $\bG^\top \bJ \bG = \bzero$ and $\rank(\bG) = \kappa$.
    Then there exists a full-rank matrix $\bG^\perp \in \Fq^{2N \times 2N - \kappa}$ such that 
    \begin{enumerate}
        \item $\bG^\top \bJ \bG^\perp = \bzero$, \ie $\symp{\bg_i}{\bg'_j} = 0$ for $i \in [\kappa],\ j \in [2N - \kappa]$, where $\bg_i$ is the $i^{th}$ column of $\bG$ and $\bg'_j$ is the $j^{th}$ column of $\bG^\perp$, and

        \item $\bG = \bG^\perp \ppsmatrix{\bI_\kappa \\ \bzero}$, \ie $\bG$ is the leftmost submatrix of $\bG^\perp$.
    \end{enumerate}
\end{lemma}

\begin{proof}
    Let $\bF$ be a symplectic completion of $\bG$, \ie a symplectic matrix such that its first $\kappa$ columns are equal to $\bG$, which exists by the conditions imposed on $\bG$ (cf. Remark~\ref{rem:sso_to_symplectic}).
    Then we can write it as $\bF = \ppmatrix{\bG & \bG_2 & \bH_1 & \bH_2}$, where $\bH_1 \in \Fq^{2N \times \kappa},\ \bG_2,\bH_2 \in \Fq^{2N \times N - \kappa}$.
    Since $\bF^\top \bJ \bF = \bJ$, it is clear that $\bff_i^\top \bJ \bff_j = 0$ for each pair of columns $\bff_i,\bff_j$ of $\bF$ such that $i \in [\kappa],\ j \in [N] \cup \{ N+\kappa+1,\ldots,2N \}$.
    Choosing $\bG^\perp = \ppmatrix{\bG & \bG_2 & \bH_2}$ proves the statement.
\end{proof}

\begin{proposition}
    \label{prop:measurementoutcome}
    Let $\bG \in \Fq^{2N \times \kappa}$ and $\bG^\perp \in \Fq^{2N \times 2N - \kappa}$ be such that 
    \begin{itemize}
        \item $\bG = \bG^\perp \ppsmatrix{\bI_\kappa \\ \bzero}$,
    
        \item $\bG^\top \bJ \bG^\perp = \bzero$,

        \item there exists $\bH \in \Fq^{2N \times \kappa}$ such that $\ppmatrix{\bG^\perp & \bH}$ is full rank.
    \end{itemize}
    Let $\cV = \colspan{\bG}$ and $(\cdot)_h: \Fq^{2N} \to \Fq^\kappa$ be such that 
    \[
    (\bx)_h \coloneqq \ppmatrix{\bzero_{\kappa \times 2N - \kappa} & \bI_\kappa} \ppmatrix{\bG^\perp & \bH}^{-1} \bx.
    \]
    Then performing the PVM $\px*{\bP^{\cV}_{\overline{\bs}} : \overline{\bs}\in\Fq^{2N}/\cV^{\symperp}}$ on the state $| \overline{\bx} \rangle$ gives the outcome $(\bx)_h$.
\end{proposition}

\begin{proof}
    Condition 3 ensures that the matrices $\bG$ and $\bG^\perp$ are full rank, so $\dim(\cV) = \kappa$ and $\dim(\cV^{\symperp}) = 2N - \kappa$.
    By conditions 1 and 2 we have that $\cV \subseteq \cV^{\symperp} = \colspan{\bG^\perp}$, so the subgroup $\cS(\cV)$ as in Equation~\eqref{eq:stab} is a stabilizer.
    By Equation~\eqref{eq:Hdecomposition} we have that $\cH^{\otimes N} = \cW \otimes \C^{q^{N-\kappa}}$, where $\cW$ is the $q^\kappa$-dimensional Hilbert space spanned by $\px*{ |\overline{\bs}\rangle : \overline{\bs}\in \Fq^{2N} / \cV^{\symperp} }$.

    Let $\bx \in \Fq^{2N}$, then we can uniquely decompose it as $\bx = \bG^\perp \bx_g + \bH \bx_h, \ \bx_g \in \Fq^{2N-\kappa},\ \bx_h \in \Fq^\kappa$.
    Notice now that $(\bx)_h = \bx_h$.
    Let $\bg \in \cV^{\symperp}$, then we can write $\bg = \bG \bg'$ for $\bg' \in \Fq^{2N-\kappa}$ and $\bx + \bg$ can be decomposed as $\bx + \bg = \bG (\bx_g + \bg') + \bH \bx_h = \bG \bx'_g + \bH \bx_h,\ \bx'_g \in \Fq^\kappa$. 

    In Proposition~\ref{prop:stab}, we have that $\overline{\bs} = \bs + \cV^{\symperp} = \px*{\bs + \bg : \bg \in \cV^{\symperp}}$.
    It follows that every $\bx \in \overline{\bs}$ maps to a unique element $\bs_h \in \Fq^\kappa$.
    Thus, we can identify each coset $\overline{\bs}$ with the element $\bs_h$.
    Then we identify the states
    \begin{equation}
        \label{eq:idwstates}
        | \bs_h \rangle_\cW = | \overline{\bs} \rangle
    \end{equation}
    to avoid confusion with the computational basis, since $\cW$ is the space spanned by the states $| \overline{\bs} \rangle$.
    
    In the decomposition given by Equation~\eqref{eq:Hdecomposition} we have $q^\kappa$ distinct elements, since each $\cH^\cV_{\overline{\bs}}$ has dimension $q^{N-\kappa}$.
    Thus, since there are $q^\kappa$ vectors $\bs_h \in \Fq^\kappa$, we can write Equation~\eqref{eq:weyldecomp} as $\bW(\bv) = \sum_{\bs_h \in \Fq^\kappa} \omega^{\symp{\bv}{\bH_1 \bs_h}} \bP^\cV_{\bs_h}$, where $\bP^\cV_{\bs_h}$ is the projection associated with the measurement outcome $\bs_h$.
    If $\kappa=N$, then $\dim(\imag{\bP^\cV_{\bs_h}}) = 1$ and we can decompose it as $\bP^\cV_{\bs_h} \coloneqq | \bs_h \rangle_\cW \langle \bs_h |_\cW$, otherwise the projection is given by a density matrix.
    
    In general, assume the system is in the state $|\overline{\bx} \rangle$ for some $\overline{\bx} \in \Fq^{2N} / \cV^\symperp$.
    By the discussion above, we can identify $\overline{\bx}$ with a unique $\bx_h \in \Fq^N$.
    Then,  
    \[
    \bP^\cV_{\bs_h} |\overline{\bx} \rangle = 
    \begin{cases}
        | \bx_h \rangle_\cW & \text{if } \bs_h = \bx_h, \\
        0 & \text{if } \bs_h \neq \bx_h.
    \end{cases}
    \]
    We thus obtain the outcome $(\bx)_h$ with probability 1 after performing the PVM $\px*{\bP^{\cV}_{\overline{\bs}} : \overline{\bs}\in\Fq^{2\kappa}/\cV^{\symperp}} = \px*{\bP^{\cV}_{\bs_h} : \bs_h \in \Fq^\kappa}$ on the state $| \overline{\bx} \rangle$.
\end{proof}

\begin{remark}
    The PVM can be more clearly expressed as 
    \begin{equation}
        \label{eq:pvm}
        \cP^\cV \coloneqq \px*{\bP^{\cV}_{\bs_h} = | \bs_h \rangle_\cW \langle \bs_h |_\cW : \bs_h \in \Fq^\kappa}.
    \end{equation}
\end{remark}

\begin{remark}
    If $\kappa = N$, then $\bG^\perp = \bG$, so the first two conditions of Proposition~\ref{prop:measurementoutcome} can be simply rewritten as $\bG^\top \bJ \bG = \bzero$.
    In this case $\bG$ defines an SSO subspace $\cV$, which is in correspondence with the maximal stabilizer $\cS(\cV)$.
    Furthermore, measuring over the PVM $\cP^\cV$ is equivalent to first revert the unitary $\bU_{\bG,\bH}$ (cf. Remark~\ref{rem:symplectic}) and measuring on the computational basis, as such unitary is needed to map the computational basis to the PVM basis.
\end{remark}

As the PVM $\cP^\cV$ is applied on $N$ qudits, one should expect $N$ $q$-ary digits as output, but if $\kappa < N$, the output of the PVM has only $\kappa$ $q$-ary digits according to Proposition~\ref{prop:measurementoutcome}.
We will clarify this aspect in Section~\ref{sec:non_max_stab_based_constructions}.

\section{$N$-Sum Box}
\label{sec:n_sum_box}

An $N$-sum box is a black box with the following functional form:
\[
\ppmatrix{y_1 \\ y_2 \\ \vdots \\ y_N} = 
\ppmatrix{M_{1,1} & \cdots & M_{1,2N} \\
          \vdots  & \ddots & \vdots \\
          M_{N,1} & \cdots & M_{N,2N}}
\ppmatrix{x_1 \\ \vdots \\ x_N \\ x_{N+1} \\ \vdots \\ x_{2N}},
\]
or equivalently $\by = \bM \bx$, where $\by \in \Fq^N$ is the output vector, $\bx \in \Fq^{2N}$ is the input vector, and $\bM \in \Fq^{N \times 2N}$ is the \emph{transfer matrix}.
The inputs to the $N$-sum box are controlled by $N$ parties (\emph{transmitters}), where transmitter $n \in [N]$ is controlling $(x_n,x_{N+n})$. 
The output vector $\by$ is measured by another party, which we label as the \emph{receiver}. 
The $N$-sum box is initialized with shared quantum entanglement among the $N$ transmitters, \ie $N$ entangled $q$-dimensional qudits are prepared and distributed to the transmitters, one qudit per transmitter. 
The initial qudit entanglement is independent of the inputs $\bx$ and any data that subsequently becomes available to the transmitters. 
No quantum resource is initially available to the receiver. 
In the course of operation of the $N$-sum box, each of the $N$ transmitters acquires data from various sources, including possibly the receiver (\eg queries in private information retrieval), based on which it performs conditional $\sX,\sZ$-gate operations on its own qudit, and then sends its qudit to the receiver.
The receiver performs a quantum measurement on the $N$ qudits, from which it recovers $\by$.

In this setting, we allow the inputs $(x_n,x_{n+N})$ from each transmitter $n \in [N]$ to be transformed by an invertible matrix.
This corresponds to multiplying the input vector by local invertible transformations, which are defined as follows.
\begin{definition}
    \label{def:lit}
    Let $\diagonal_{N,\Fq}$ be the set of diagonal matrices of dimension $N \times N$ and entries in $\Fq$. The set of {\em local invertible transformations (LITs)} is defined as
    \[ \begin{split}
        \LIT_{N,\Fq} \coloneqq \px*{ \ppmatrix{\bgL_1 & \bgL_2 \\ \bgL_3 & \bgL_4} : \bgL_i \in \diagonal_{N,\Fq},\ i \in [4],\ \det(\bgL_1 \bgL_4 - \bgL_2 \bgL_3) \neq 0 }.
    \end{split} \]
\end{definition}
Notice that the submatrix with the entries in position $(n,n)$, $(n,n+N)$, $(n+N,n)$, $(n+N,n+N)$ of $\bgL \in \LIT_{N,\Fq}$ is the invertible matrix applied by transmitter $n \in [N]$ to its inputs.

We also allow \emph{receiver invertible transformations}, \ie we allow the receiver to transform the output vector of the $N$-sum box by multiplying it by $\bP \in \GL_{N,\Fq}$, where $\GL_{N,\Fq}$ is the set of invertible matrices with dimension $N$ and entries in $\Fq$.
This gives equivalent representations of the $N$-sum box as
\[
\by = \bP \bM \bgL \bx, \quad \bP \in \GL_{N,\Fq},\ \bgL \in \LIT_{N,\Fq}.
\]

\begin{definition}
    The relation $\liteq$ defines an equivalence class of pairs of matrices $\bM_1, \bM_2 \in \Fq^{N \times 2N}$ up to local and receiver invertible transformations, \ie
    \[
    \bM_1 \liteq \bM_2 \iff \bM_1 = \bP \bM_2 \bgL,
    \]
    where $\bP \in \GL_{N,\Fq},\ \bgL \in \LIT_{N,\Fq}$.
\end{definition}

\section{Stabilizer-based $N$-Sum Boxes}
\label{sec:sbnsumbox}

\begin{figure*}[ht]
\centering
\begin{tikzpicture}

\node at (0.3,-2) [align=center,rotate=90] {\footnotesize $N$ qudits in entangled state $\ket{\bzero}_\cW $\\ \footnotesize $ \ket{\bzero}_\cW= \bU_{\bG,\bH}\ket{\bzero}$};


\begin{scope}[shift={(1.0,0)}]

\node (U1) at (1.75,-0.15) [ help lines, rectangle, minimum width=2.6cm, minimum height=1cm, fill=red!10] {};
\node (U2) at (1.75,-1.9) [ help lines, rectangle, minimum width=2.6cm, minimum height=1cm, fill=blue!10] {};
\node (U3) at (1.75,-4.15) [ help lines, rectangle, minimum width=2.6cm, minimum height=1cm, fill=teal!10] {};
\node [below=-0.1cm of U1] {\tiny{Tx1} with data $(x_1,x_{1+N})$};
\node [below=-0.1cm of U2] {\tiny{Tx2} with data $(x_2,x_{2+N})$};
\node [below=-0.1cm of U3] {\tiny{Tx$N$} with data $(x_N,x_{2N})$};
\node (dd) [below = 0.3cm of U2] {$\vdots$};

\node (Z1) at (1.1,0) [draw, rectangle, inner sep =0.1cm] {$\sZ$};
\node (Z2) at (1.1,-1.75) [draw, rectangle, inner sep =0.1cm] {$\sZ$};
\node (Z3) at (1.1,-4) [draw, rectangle, inner sep =0.1cm] {$\sZ$};
\node (X1) at (2.4,0) [draw, rectangle, inner sep =0.1cm] {$\sX$};
\node (X2) at (2.4,-1.75) [draw, rectangle, inner sep =0.1cm] {$\sX$};
\node (X3) at (2.4,-4) [draw, rectangle, inner sep =0.1cm] {$\sX$};
\node (a2) at (X1) [below=0.2cm] {\footnotesize $x_1$};
\node (a1) at (Z1) [below=0.2cm] {\footnotesize $x_{1+N}$};
\node (b2) at (X2) [below=0.2cm] {\footnotesize $x_2$};
\node (b1) at (Z2) [below=0.2cm] {\footnotesize $x_{2+N}$};
\node (c2) at (X3) [below=0.2cm] {\footnotesize $x_N$};
\node (c1) at (Z3) [below=0.2cm] {\footnotesize $x_{2N}$};

\draw [color=black, thick] ($(Z1.west) - (0.75,0)$)--(Z1)--(X1)--($(X1.east) + (0.75,0)$);
\draw [color=black, thick] ($(Z2.west) - (0.75,0)$)--(Z2)--(X2)--($(X2.east) + (0.75,0)$);
\draw [color=black, thick] ($(Z3.west) - (0.75,0)$)--(Z3)--(X3)--($(X3.east) + (0.75,0)$);

\node (Mid) at (1.75,-2.25) [  help lines, rectangle, minimum width=3.0cm, minimum height=5.5cm] {};
\node [above=0cm of Mid]{\footnotesize Transmitters};

\node (L2) [below=0cm of Mid, align=center]{\footnotesize $\bx = \bG \bx_g + \bH \bx_h $,\\[-0.05cm] \footnotesize $\bx \in \Fq^{2N} \mbox{ and }\bx_g, \bx_h \in \Fq^N$};

\end{scope}

\begin{scope}[shift={(-5.0,0)}]
\node (D) at (13.3,-2) [ help lines, fill =orange!10, rectangle, minimum width=4.3cm, minimum height=5cm] {};
\node [below=0cm of D, align=center]{\footnotesize { Rx} computes $\by = \bx_h = \bM \bx$,\\[0.1cm]\footnotesize $\bM = \ppmatrix{\bzero & \bI}\ppmatrix{\bG & \bH}^{-1}$,\\ \footnotesize \mbox{$\bM\in \Fq^{N\times 2N}$}};
\node (R2) at (12,-2) [draw, very thick, rectangle, minimum width=1.0cm, minimum height=4.5cm, fill=white] {$\bU_{\mbox{\tiny $\bG,\bH$}}^{\dagger}$};
\draw [thick] ($(X1-|R2.west) - (0.75, 0)$)--(Z1-|R2.west);
\draw [thick] ($(X2-|R2.west) - (0.75, 0)$)--(Z2-|R2.west);
\draw [thick] ($(X3-|R2.west) - (0.75, 0)$)--(Z3-|R2.west);
\end{scope}

\draw [color=black, dashed] ($(X1.east) + (0.75,0)$)--($(X1-|R2.west) - (0.75, 0)$) node [midway, align=center] {\footnotesize $1$ qudit\\ \footnotesize to Rx};
\draw [color=black, dashed] ($(X2.east) + (0.75,0)$)--($(X2-|R2.west) - (0.75, 0)$) node [midway, align=center] {\footnotesize $1$ qudit\\ \footnotesize to Rx};
\draw [color=black, dashed] ($(X3.east) + (0.75,0)$)--($(X3-|R2.west) - (0.75, 0)$) node [midway, align=center] {\footnotesize $1$ qudit\\ \footnotesize to Rx};

\begin{scope}[shift={(-4.3,0)}]
\node (mx1) at (12.5,0) [meter] {};
\node (mx2) at (12.5,-1.75) [meter] {};
\node [below = 0.5cm of mx2] {$\vdots$};
\node (mx3) at (12.5,-4) [meter] {};

\draw [thick] (R2.east|-mx1)--(mx1);
\draw [thick] (R2.east|-mx2)--(mx2);
\draw [thick] (R2.east|-mx3)--(mx3);

\node [above=0cm of D]{\footnotesize Receiver};

\node (O1) at (13,0) [anchor=west]{\footnotesize $y_1 = (\bx_h)_1$} ;
\node (O2) at (13,-1.75) [anchor=west]{\footnotesize $y_2 = (\bx_h)_2$} ;
\node [below = 0.5cm of O2] {$\vdots$};
\node (O3) at (13,-4) [anchor=west]{\footnotesize $y_N = (\bx_h)_N$} ;

\draw [double] (mx1)--(O1);
\draw [double] (mx2)--(O2);
\draw [double] (mx3)--(O3);
\end{scope}

\node (box) at (8.55,-2.5) [rectangle, help lines, draw, minimum width=18.0cm, minimum height=8cm]{};
\draw [] ($(box.north)+(2.4,0)$) -- ($(box.south)+(2.4,0)$);

\begin{scope}[shift={(12,0)}]

\node (Tx1) at (0,-0.15) [help lines, text =black, rectangle, minimum height=0.9cm, minimum width=1cm, fill=red!10] {};
\node (Tx2) at (0,-1.9) [help lines, text =black, rectangle, minimum height=0.9cm, minimum width=1cm, fill=blue!10] {};
\node (Tx3) at (0,-4.15) [help lines, text =black, rectangle, minimum height=0.9cm, minimum width=1cm, fill=teal!10] {};
\node [below=-0.1cm of Tx1] {\tiny \bf Tx1};
\node [below=-0.1cm of Tx2] {\tiny \bf Tx2};
\node [below=-0.1cm of Tx3] {\tiny \bf Tx$N$};
\node (vd) [below = 0.3cm of Tx2] {$\vdots$};

\node [left=6pt of {$(Tx1.east)+(0,0.25)$}] {\footnotesize $x_1$};
\node [left=0pt of {$(Tx1.east)-(0,0.25)$}] {\footnotesize $x_{1+N}$};

\node [left=6pt of {$(Tx2.east)+(0,0.25)$}] {\footnotesize $x_2$};
\node [left=0pt of {$(Tx2.east)-(0,0.25)$}] {\footnotesize $x_{2+N}$};

\node [left=4pt of {$(Tx3.east)+(0,0.25)$}] {\footnotesize $x_N$};
\node [left=2pt of {$(Tx3.east)-(0,0.25)$}] {\footnotesize $x_{2N}$};

\node (R) at (4.8,-2) [ help lines, rectangle, text=black, fill=orange!10, minimum width=0.8cm, minimum height=3cm, align = center]{};
\node [right=0.08cm of {$(R.west)+(0,1)$}] {\footnotesize $y_1$};
\node [right=0.08cm of {$(R.west)+(0,0.25)$}] {\footnotesize $y_2$};
\node [right=0.08cm of {$(R.west)-(0,1)$}] {\footnotesize $y_N$};
\node [right=0.2cm of {$(R.west)-(0,0.375)$}] {\footnotesize $\vdots$};
\node [left=0.2cm of {$(R.west)-(0,0.375)$}] {\footnotesize $\vdots$};

\node [above=0cm of R, align=center]{\footnotesize Receiver};

\node (G) at (2.5,-2.15) [draw, very thick, rectangle, minimum width=2.2cm, minimum height=4.8cm, fill=black!10, align=center] { \footnotesize $N$-Sum Box\\ [1.1cm] $\by = \bM \bx$\\[0.55cm]~\\  \footnotesize Com. Cost \\[-0.1cm] \footnotesize $N$ qudits};

\node [below=0.5cm of G, align=center] {\footnotesize $\by=(y_1,y_2,\cdots,y_N)^\top$\\ \footnotesize $\bx=(x_1,x_2,\cdots,x_{2N})^\top$};

\draw [thick, ->] let \p{Geast}=($(G.east)+(0,1.1)$), \p{Rwest}=(R.west) in (\x{Geast},\y{Geast})--(\x{Rwest},\y{Geast}) node [above, midway] {};
\draw [thick, ->] let \p{Geast}=($(G.east)-(0,0.9)$), \p{Rwest}=(R.west) in (\x{Geast},\y{Geast})--(\x{Rwest},\y{Geast}) node [above, midway] {};
\draw [thick, ->] let \p{Geast}=($(G.east)+(0,0.37)$), \p{Rwest}=(R.west) in (\x{Geast},\y{Geast})--(\x{Rwest},\y{Geast}) node [below, midway] {};

\draw [thick, ->] let \p{Tx1east}=($(Tx1.east)+(0,0.25)$), \p{Gwest}=(G.west) in (\x{Tx1east},\y{Tx1east})--(\x{Gwest},\y{Tx1east}) node [above, midway] {}; 
\draw [thick,->] let \p{Tx1east}=($(Tx1.east)-(0,0.25)$), \p{Gwest}=(G.west) in (\x{Tx1east},\y{Tx1east})--(\x{Gwest},\y{Tx1east}) node [above, midway] {}; 
\draw [thick, ->] let \p{Tx2east}=($(Tx2.east)+(0,0.25)$), \p{Gwest}=(G.west) in (\x{Tx2east},\y{Tx2east})--(\x{Gwest},\y{Tx2east}) node [below, midway] {};
\draw [thick, ->] let \p{Tx2east}=($(Tx2.east)-(0,0.25)$), \p{Gwest}=(G.west) in (\x{Tx2east},\y{Tx2east})--(\x{Gwest},\y{Tx2east}) node [below, midway]{};
\draw [thick, ->] let \p{Tx3east}=($(Tx3.east)+(0,0.25)$), \p{Gwest}=(G.west) in (\x{Tx3east},\y{Tx3east})--(\x{Gwest},\y{Tx3east}) node [below, midway] {};
\draw [thick, ->] let \p{Tx3east}=($(Tx3.east)-(0,0.25)$), \p{Gwest}=(G.west) in (\x{Tx3east},\y{Tx3east})--(\x{Gwest},\y{Tx3east}) node [below, midway]{};

\end{scope}
\end{tikzpicture}
\caption{Quantum circuit and black-box representation for an $N$-sum box with transfer function $\by = \bM \bx$.}
\label{fig:sbcNsumbox}
\end{figure*}
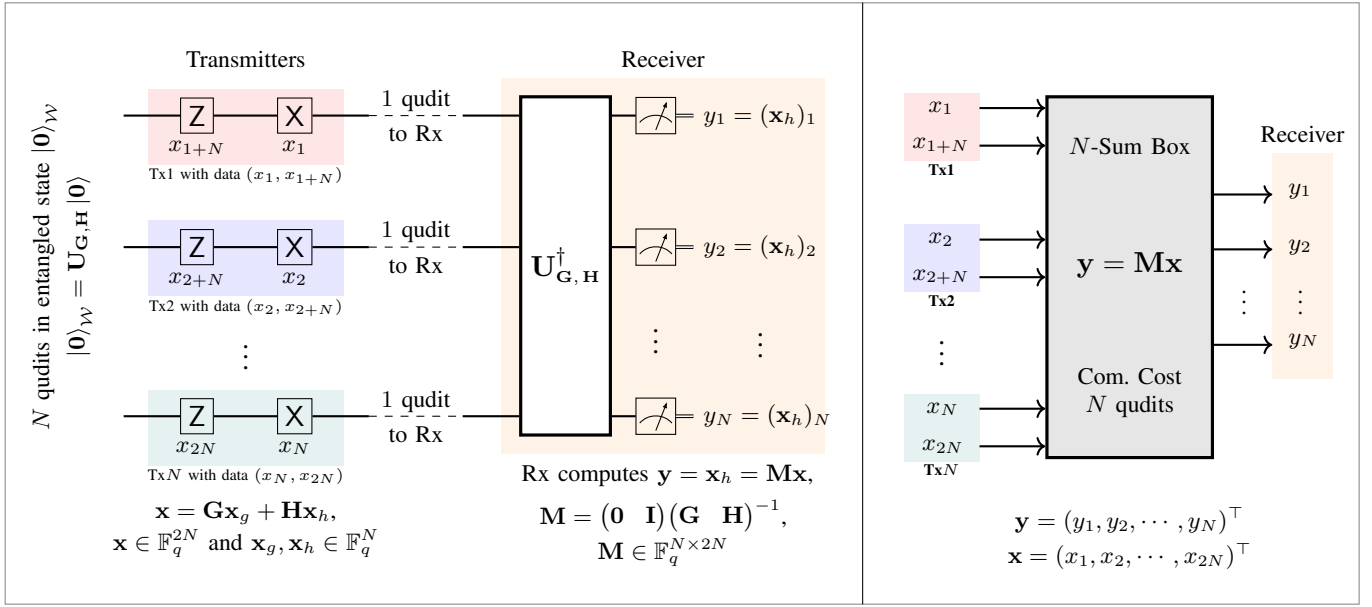

First, we define strongly self-orthogonal matrices, which are used in the construction of $N$-sum boxes based on the stabilizer formalism.

\begin{definition}
    \label{def:ssomatrix}
    A matrix $\bM \in \Fq^{2N \times N}$ is said to be {\em strongly self-orthogonal (SSO)} if its columns span an SSO subspace, or equivalently, if $\bM^\top \bJ \bM = \bzero$ and $\rank(\bM) = N$.
    The set of SSO matrices is denoted by $\cM_o$.
\end{definition}

\begin{remark}
    The reason why we define strongly self-orthogonal matrices is because self-orthogonal matrices would not generate SSO subspaces.
    For example, in $\F_9^4$ where $\F_9 \equiv \F_3[x] / (x^2+x+2)$ with generator element $\a$, the matrix
    \[
    \bG^\top = \ppmatrix{1 & 0 & 0 & 2 \\ 0 & 1 & -\a & 0}
    \]
    is self-orthogonal since the $\F_3$-trace of each element of the matrix
    \[
    \bG^\top \bJ \bG = \ppmatrix{0 & \a+2 \\ \a+2 & 0}
    \]
    is 0, but the space spanned by its rows wouldn not be strongly self-orthogonal since
    \[
    \symp{\a\bv_1}{\bv_2} = \trace_{\F_9 / \F_3}(2\a + \a^2) \neq 0. 
    \]
\end{remark}

Let us now characterize some classes of $N$-sum boxes that can be constructed based on the stabilizer formalism.

\subsection{Case with disallowed LITs}

The following theorem describes which transfer matrices are feasible from a stabilizer-based construction when LITs are disallowed.

\begin{theorem}
    \label{thm:nsumboxconstruction}
    Suppose there exists $\bG \in \Fq^{2N \times N}$ such that
    \begin{enumerate}
        \item \label{item:Gsso} $\bG \in \cM_o$,
        \item there exists $\bH \in \Fq^{2N \times N}$ such that $\ppmatrix{\bG & \bH}$ is full-rank, 
        \item \label{item:HGtransfermatrix} $\bM \in \Fq^{N \times 2N}$ is the submatrix comprised of the bottom $N$ rows of $\ppmatrix{\bG & \bH}^{-1}$, \ie $\bM \coloneqq \ppmatrix{\bzero & \bI} \ppmatrix{\bG & \bH}^{-1}$.
    \end{enumerate}
    Then there exists a stabilizer-based construction for an $N$-sum box over $\Fq$ with transfer matrix $\bM$.
\end{theorem}

\begin{proof}
    Let $\bU_{\bG,\bH} \in \C^{q^N \times q^N}$ be the unitary matrix such that its $i^{th}$ column is the vector representing the state $| \nu(i) \rangle_\cW$, $i \in [q^N]$, identified by Equation~\eqref{eq:idwstates}, where $\nu: [q^N] \to \Fq^N$ is the natural isomorphism (cf. Remark~\ref{rem:symplectic}).
    Let $\zerostate_\cW = \bU_{\bG,\bH} \zerostate$ be the initial entangled state over $\cH^{\otimes N}$.
    Let $\cV = \colspan{\bG}$.
    Assume that transmitter $n \in [N]$ applies $\sX(x_n), \sZ(x_{N+n})$ on his qudit and sends it to the receiver. 
    Then the quantum system received is in the state $\bW(\bx) \zerostate_\cW = | (\bx)_h \rangle_\cW$.
    After performing the PVM $\cP^\cV$ defined in Equation~\eqref{eq:pvm} on the qudits the receiver measures $(\bx)_h$ without error by Proposition~\ref{prop:measurementoutcome}.
    Let $\bM$ be the submatrix comprised of the bottom $N$ rows of $\ppmatrix{\bG & \bH}^{-1}$, then we have that $\bM \bx = \bx_h$, which is the output of the measurement. 
    We proved that for an $N$-sum box with transfer function $\bM$ satisfying condition~\ref{item:HGtransfermatrix} there exists a quantum black box with input $\bx$ and output $\bM\bx$.
\end{proof}

\begin{remark}    
\label{rem:nsumboxconstruction2}
    The terminology ``stabilizer-based construction" stems from the aforementioned correspondence between stabilizers and self-orthogonal spaces. 
    Explicitly, let $\cS = \langle \bW(\bs_1),\ldots,\bW(\bs_\kappa) \rangle \subseteq \HWqN$ be a stabilizer group, \ie a stabilizer group with $\kappa$ independent generators $\bW(\bs_i)$ dependent on $\bs_i \in \Fq^{2N},\ i \in [\kappa]$.
    Let $\bG \in \Fq^{2N \times \kappa}$ be the matrix that has $\bs_i$ as its $i^{th}$ column, then $\bG^\top \bJ \bG = \bzero$ and $\rank(\bG) = \kappa$.
    For the case $\kappa=N$, the stabilizer is maximal, $\zerostate_\cW$ is its stabilized state, and $\bG \in \cM_o$.
\end{remark}

\begin{remark}
    A stabilizer-based construction for any feasible $N$-sum box $\by = \bM \bx$ is information-theoretically optimal as a black-box implementation in the sense that it has the least possible quantum download cost. 
    In other words, since the transfer matrix $\bM$ is full rank, there cannot exist a more efficient (in terms of download cost) construction of the same $N$-sum box by some other (non-stabilizer-based) means so that the output $\by$ is capable of delivering $N$ $q$-ary digits to the receiver, which cannot be done with a communication cost of less than $N$ qudits by the Holevo bound~\cite{holevo1973bounds}.
\end{remark}

\begin{remark}
    \label{rem:quantumtrivialprotocol}
    From a quantum coding-theoretic perspective, a stabilizer-based construction for an $N$-sum box is equivalent to preparing a stabilizer state, applying an $N$-qudit error corresponding to a string of $2N$ symbols from a finite field, and computing the syndrome on this state as output.
\end{remark}

We denote by $\cM_{sbc}$ the set of all the transfer matrices resulting from stabilizer-based constructions with disallowed LITs. 
In the following, we establish that the set of transfer matrices achievable through the $(\bG,\bH)$-construction established by Theorem~\ref{thm:nsumboxconstruction} is the same as the set of SSO matrices, \ie

\begin{lemma}
    \label{lemm:charM}
    $\cM_o = \cM_{sbc}$.
\end{lemma}

\begin{proof}
    Let $\bM \in \Fq^{N \times 2N}$ be such that $\bM^\top \in \cM_{sbc}$.
    By condition~\ref{item:HGtransfermatrix} of Theorem~\ref{thm:nsumboxconstruction} we have that $\bM \bG = \bzero = \bG^\top \bJ \bG$.
    Since $\bG$ is full-rank, also $\bG^\top \bJ$ is full-rank, so $\rowspan{\bM} = \rowspan{\bG^\top \bJ}$.
    This implies that $\bM = \bP \bG^\top \bJ$ for an invertible matrix $\bP \in \Fq^N$, and since $\bJ\bJ^\top = \bI$ we can conclude that $\bM \bJ \bM^\top = \bP \bG^\top \bJ \bJ \bJ^\top \bG \bP^\top = \bzero$, \ie $\bM^\top \in \cM_o$ and $\cM_{sbc} \subseteq \cM_o$.
    
    Now, let $\bM \in \Fq^{N \times 2N}$ be such that $\bM^\top \in \cM_o$. 
    Let $\bN \in \Fq^{N \times 2N}$ be such that $\ppmatrix{\bN^\top & \bM^\top}$ is full-rank.
    Then the matrix $\ppmatrix{\bN^\top & \bM^\top}$ is invertible and we can write its inverse as $\ppmatrix{\bG & \bH}$, where $\bG,\bH \in \Fq^{N \times 2N}$ are full-rank matrices. 
    Clearly, $\bM \ppmatrix{\bG & \bH} = \ppmatrix{\bzero & \bI}$, so by the same argument above it is easy to see that $\bG \in \cM_o$, \ie $\cM_o \subseteq \cM_{sbc}$.
\end{proof}

\begin{remark}
    \label{rem:symplectic}
    Let $\bG = \ppsmatrix{\bA \\ \bB} \in \cM_o$ be any SSO matrix and $\bH = \ppsmatrix{\bC \\ \bD}$ be such that the matrix $\bF \coloneqq \ppmatrix{\bG & \bH}$ is symplectic.
    Then, by Equation~\eqref{eq:SympInv}, we have $\ppmatrix{\bzero & \bI} \bF^{-1} = \ppmatrix{-\bB^\top & \bA^\top}$ which is again an SSO matrix. 
    Lemma~\ref{lemm:charM} implies that we only need to complete $\bG$ to a symplectic matrix (instead of invertible).
    The well-known connection between symplectic matrices and the stabilizer formalism~\cite{CRSS98} allows for a simpler description of the matrix $\bU_{\bG,\bH}$, as the following example illustrates.
\end{remark}

\begin{example}
    \label{ex:two_sum_box}
    Suppose we have two parties, Tx1 and Tx2, both possessing two bits $\ba = (x_1, x_3),\ \bb = (x_2, x_4)\in \F_2^2$, respectively (see Figure~\ref{fig:twosumbox}). 
    The two-sum transmission protocol computes the sum of their bits $(x_1 + x_2, x_3 + x_4)$ starting with the Bell state $|\beta_{00}\rangle = (|00\rangle+|11\rangle)/\sqrt{2}$, which is stabilized by the stabilizer $\cS = \langle \bW(0,0,1,1),\bW(1,1,0,0)\rangle$.
    Consider the matrix 
    \begin{equation*}
    \bF = \ppmatrix{\bG & \bH} = 
    \ppmatrix{0 & 1 & 1 & 0 \\ 
              0 & 1 & 0 & 0 \\ 
              1 & 0 & 0 & 0 \\ 
              1 & 0 & 0 & 1},
    \end{equation*}
    where $\bG$ is determined by $\cS$ and $\bH$ is chosen so that $\bF$ is a symplectic matrix. The symplectic matrix can be decomposed, \eg using the Bruhat decomposition~\cite{MR18,PTC22}, as 
    \begin{equation*}
    \bF = 
    \ppmatrix{1 & 1 & 0 & 0 \\
              0 & 1 & 0 & 0 \\
              0 & 0 & 1 & 0 \\
              0 & 0 & 1 & 1} \cdot 
    \ppmatrix{0 & 0 & 1 & 0 \\
              0 & 1 & 0 & 0 \\
              1 & 0 & 0 & 0 \\
              0 & 0 & 0 & 1}.
    \end{equation*}
    The components correspond to quantum gates $\mathsf{CNOT}$ and a partial Hadamard $\sH \otimes \sI$ on the first qudit~\cite{PRTC20}. This is precisely the circuit $\bU_{\bG,\bH}$ for which $|\beta_{00}\rangle = \bU_{\bG,\bH}|00\rangle$.
    Using Equation~\eqref{eq:SympInv} we then obtain the transfer matrix
    \begin{equation*}
    \bM = \ppmatrix{{\bf 0} & \bI}\bF^{-1} = \ppmatrix{1&1&0&0\\0&0&1&1},
    \end{equation*}
    which is exactly the functional form of the two-sum transmission protocol.
    This approach allows for a straightforward generalization that computes $\mathbf{M}\bx$ starting with the state $(|0..0\rangle + |1..1\rangle) / \sqrt{2}$, where
    \[
    \mathbf{M} = \left(\!\!\begin{array}{c|c}
    110 \cdots 00 & 000 \cdots 00 \\
    011 \cdots 00 & 000 \cdots 00 \\ 
    \vdots & \vdots \\
    000 \cdots 11 & 000 \cdots 00 \\
    \hline
    000 \cdots 00 & 111 \cdots 11
    \end{array}\!\!\right)\in \Fq^{N\times 2N}.
    \]
\end{example}

The following theorem fully characterizes $N$-sum boxes without LITs and follows directly from Theorem~\ref{thm:nsumboxconstruction}, Remark~\ref{rem:nsumboxconstruction2} and Lemma~\ref{lemm:charM}.

\begin{theorem}
    Let $\bM \in \cM_o$. A construction based on a stabilizer $\cS \subseteq \HWqN$ exists for an $N$-sum box over $\Fq$ with transfer matrix $\bM^\top$ if and only if $\cS$ is a maximal stabilizer.
\end{theorem}


\subsection{Case with allowed LITs}

Now we explore what is possible when LITs are allowed. 
The following theorem shows that the transfer matrix of any $(\bG,\bH)$ construction is equivalent (up to LITs) to the null-space of $\bG$, which can be explicitly represented as $\bJ^\top \bG$.

\begin{proposition}
    Let $\bM \in \Fq^{N \times 2N}$ be such that $\bM^\top \in \cM_{sbc}$. Then $\bM \liteq \bG^\top \bJ$, the null-space of $\bG^\top$.
\end{proposition}

\begin{proof}
    This follows directly from the proof of Lemma~\ref{lemm:charM}, since we can write $\bM = \bP \bG^\top \bJ$
    for $\bP \in \GL_{N,\Fq}$, \ie $\bM \liteq \bJ^\top \bG$.
\end{proof}

A known property of SSO matrices is that they can be written in the so-called \emph{standard form}~\cite{Omeara}, \ie if $\bM \in \cM_o$ then there exist $\bP \in \GL_{N,\Fq},\ \bQ \in \GL_{2N,\Fq}$ such that
\begin{equation}
    \label{eq:standardformnonlit}
    \bP \bM^\top \bQ = \ppmatrix{\bI & \bS},
\end{equation}
where $\bS$ is a symmetric $N \times N$ matrix, \ie $\bS^\top = \bS$. 
We denote by $\cM_s$ the set of transfer matrices in standard form, \ie
\[ \begin{split}
    \cM_s \coloneqq \px*{ \bM \in \Fq^{2N \times N} : \bM^\top = \ppmatrix{\bI & \bS},\ \bS \in \Fq^{N \times N},\ \bS^\top = \bS }.
\end{split} \]
The following lemma shows that any SSO matrix $\bM$ can be transformed by at most $N$ signed column-swapping operations into a matrix $\bM' = \ppmatrix{\bM'_l & \bM'_r}^\top \in \cM_o$ such that $\bM'_l$ is full-rank. For completeness, the proof is included in Appendix~\ref{app:colswaplemma}.

\begin{lemma}
    \label{lemm:colswap}
    For any $\bM \in \cM_o$ there exists a diagonal matrix $\bgS \in \px*{0,1}^{N \times N}$ such that
    \begin{align}
        (\bM')^\top & = \bM^\top \ppmatrix{\bI - \bgS & \bgS \\ -\bgS & \bI - \bgS} \liteq \bM^\top, \label{eq:colswapequivalence} \\
        \bM' & \in \cM_o, \label{eq:colswapsso} \\
        \det\paren*{\bM'_l} & \neq 0. \label{eq:colswapdet}
    \end{align}
\end{lemma}

\begin{remark}
    Notice that $(\bM')^\top$ is obtained from $\bM^\top$ by \emph{signed} column-swapping operations, \ie swapping corresponding columns of $\bG_l$ and $\bG_r$ with a sign-change operation. 
    Specifically, if $\bgS_{i,i} = 1$, the $i^{th}$ column of $\bG_l$ is replaced with the \emph{negative} of the $i^{th}$ column of $\bG_r$, while the $i^{th}$ column of $\bG_r$ is replaced with the $i^{th}$ column of $\bG_l$.
\end{remark}

Our next result shows that with LITs, every feasible transfer matrix $\bM \in \cM_o$ has an equivalent representation in the standard form $\bM \liteq \ppmatrix{\bI & \bS}$, where $\bS^\top = \bS$. 
Notice that in Equation~\eqref{eq:standardformnonlit} the matrix $\bQ$ is not necessarily an LIT.
The contribution here is to show that the standard form remains valid when only LITs are allowed. 

\begin{theorem}
    \label{thm:contructiontostandardform}
    For every $\bM \in \cM_{sbc}$ there exists $\bM' \in \cM_s$ such that $(\bM')^\top \liteq \bM^\top$. Conversely, $\cM_s \subseteq \cM_{sbc}$.
\end{theorem}

\begin{proof}
    By Lemma~\ref{lemm:colswap} in Appendix~\ref{app:colswaplemma}, there exists $\bM' \coloneqq \ppmatrix{\bM'_l & \bM'_r}^\top \in \cM_o$ such that $\bM^\top \liteq (\bM')^\top$ and $\bM'_l$ is full-rank. 
    Since multiplication on the left by an invertible square matrix preserves both rank and strong self-orthogonality, we obtain that $(\bM'')^\top = (\bM'_l)^{-1} (\bM')^\top = \ppmatrix{\bI & \bF}$ and $\bM'' \in \cM_o$. Notice that $\bF$ is symmetric, since $\bM'' = \ppmatrix{\bI & \bF}^\top \in \cM_o$ implies that
    \[
    \bzero = (\bM'')^\top \bJ \bM'' = \ppmatrix{\bF & -\bI} \ppmatrix{\bI & \bF}^\top = \bF^\top - \bF,
    \]
    so we can conclude that $\bM'' \in \cM_s$.

    Conversely, it is easy to see that the standard form $\ppmatrix{\bI & \bS}$ is strongly self-orthogonal, since it has rank $N$ and
    \[
    \ppmatrix{\bI & \bS} \bJ \ppmatrix{\bI & \bS}^\top = \ppmatrix{\bS & -\bI} \ppmatrix{\bI & \bS}^\top = \bS - \bS^\top = \bzero.
    \]
    Clearly we have that $\cM_s \subseteq \cM_o = \cM_{\bG,\bH}$.
\end{proof}

\begin{remark}
    The standard form is not unique up to LITs. 
    For example,
    \[ \begin{split}
        \ppmatrix{\bI & \bS_1} & = 
        \ppmatrix{\textcolor{red}{1} & 0 & \textcolor{green}{1} & 1 \\
                  \textcolor{red}{0} & 1 & \textcolor{green}{1} & 1} \liteq
        \ppmatrix{\textcolor{green}{1} & 0 & \textcolor{red}{1} & 1 \\
                  \textcolor{green}{1} & \textcolor{orange}{-}1 & \textcolor{red}{0} & 1} \\
        & \liteq \ppmatrix{1 & 0 \\ 1 & -1}
        \ppmatrix{1 & 0 & 1 & 1 \\
                  1 & -1 & 0 & 1} \\
        & = \ppmatrix{1 & 0 & 1 & 1 \\
                      0 & 1 & 1 & 0} = \ppmatrix{\bI & \bS_2}, \quad \bS_2 \neq \bS_1.
    \end{split} \]
\end{remark}

Next, we define the set $\cM_{\LIT}$ of all possible transfer matrices that we can obtain by applying LITs:
\[
\cM_{\LIT} \coloneqq \px*{\bM \in \Fq^{2N \times N} : \exists \bM' \in \cM_{sbc},\ \bM \liteq \bM'}.
\]
The following theorem shows that any transfer matrix in $\cM_{\LIT}$ is equivalent, up to LITs, to a transfer matrix in standard form.
Figure~\ref{fig:sbc_fig} provides an overview of the relationships between the various forms of transfer functions.

\begin{theorem}
    \label{thm:littostandardform}
    For every $\bM \in \cM_{\LIT}$ there exists $\bM' \in \cM_s$ such that $\bM' \liteq \bM$. Conversely, for every $\bM' \in \cM_s$ there exists $\bM \in \cM_{\LIT}$ such that $\bM \liteq \bM'$.
\end{theorem}

\begin{proof}
    By definition, for every $\bM \in \cM_{\LIT}$ there exists $\bM' \in \cM_{sbc}$ such that $\bM' \liteq \bM$. By Theorem~\ref{thm:contructiontostandardform} there exists $\bM'' \in \cM_s$ such that $\bM'' \liteq \bM'$, and since the equivalence is transitive, we have that $\bM'' \liteq \bM$.

    The converse is trivial by the definitions.
\end{proof}

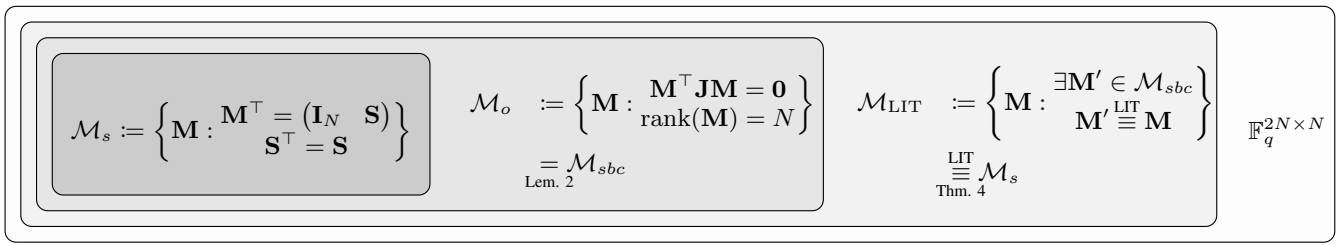
\begin{figure}
    \centering
    \resizebox{\linewidth}{!}{
    \begin{tikzpicture}

\draw [fill=black!1, opacity=1, rounded corners] (-0.4,-0.4) rectangle (16.5,2.6) ;
\node at (15.9,1){\footnotesize $\Fq^{2N\times N}$};

\draw [fill=black!5, opacity=1, rounded corners] (-0.2,-0.2) rectangle (15,2.4) ;

\node at (12.7,1)  [align=center]{ \footnotesize $\begin{array}{ll}\cM_{\LIT}&\coloneqq \px*{\bM : \begin{matrix}\exists \bM' \in \cM_{sbc} \\ \bM'\liteq \bM\end{matrix}}\\[0.4cm] & \mathrel{\underset{\makebox[0pt]{\mbox{\normalfont\tiny Thm.~\ref{thm:littostandardform}}}}{\liteq}}  \cM_s\end{array}$};

\draw [fill=black!10, opacity=1, rounded corners] (0,0) rectangle (10,2.2) ;

\node at (7.7,1)  [align=center]{  \footnotesize $\begin{array}{ll}\cM_o&\coloneqq \px*{\bM : \begin{matrix}\bM^\top \bJ \bM = \bzero\\ \rank(\bM) = N\end{matrix}}\\[0.4cm]& \mathrel{\underset{\makebox[0pt]{\mbox{\normalfont\tiny Lem.~\ref{lemm:charM}}}}{=}} \cM_{sbc}\end{array}$};

\draw [fill=black!20, opacity=1, rounded corners] (0.2,0.2) rectangle (5,2) ;

\node  at (2.6,1) [align=center] { \footnotesize $\cM_s\coloneqq \px*{\bM : \begin{matrix}\bM^\top = \ppmatrix{\bI_N & \bS}\\ \bS^\top = \bS\end{matrix}}$};

\end{tikzpicture}
    }
    \caption{Relationships between the various forms of transfer functions.}
    \label{fig:sbc_fig}
\end{figure}

\subsection{Feasibility of an $N$-Sum Box}

Next we explore the problem of testing feasibility. Given a desired $N$-sum box specification $\bY = \bM^\top \bX$ for some $\bM \in \Fq^{2N \times N}$, the goal is to determine if this $N$-sum box is feasible by stabilizer based constructions combined with local invertible transformations. 
The easiest case is if $\bM$ is in the standard form $\bM^\top = \ppmatrix{\bI & \bS}$, which is immediately seen to be feasible. 
Another easy case is when $\bM \in \cM_o$, which can be checked efficiently as well with complexity $O(N^3)$. 
However, if $\bM$ is neither in standard form, nor a
full-rank strongly self-orthogonal matrix, it could still be feasible through local invertible transformations.
Here we describe a test to determine such feasibility, or to rule it out with certainty, with complexity no more than $O(N^4)$.

\begin{theorem}
    The transfer matrix $\bM^\top = \ppmatrix{\bM_l & \bM_r} \in \Fq^{N \times 2N}$ is feasible for an $N$-sum box construction, \ie $\bM \in \cM_{\LIT}$, if and only if $\rank(\bM) = N$ and there exists an invertible diagonal matrix $\bgD \in \Fq^{N \times N}$ such that
    \begin{equation}
        (\bM')^\top = \ppmatrix{\bM_l & \bM_r \bgD} \in \cM_o.
    \end{equation}
\end{theorem}

\begin{remark}
    Finding such a $\bgD$ can be done with complexity $O(N^4)$ by solving for the $N$ elements of $\bgD$ subject to the linear constraints imposed by the self-orthogonality condition.
\end{remark}

\begin{proof}
    Consider a general feasible setting.
    By Theorem~\ref{thm:littostandardform} we can write $\bM^\top = \bP \ppmatrix{\bI & \bS} \bgL$ for some $\bP \in \GL_{N,\Fq}$ and
    \[
    \bgL = \ppmatrix{\bgL_1 & \bgL_2 \\ \bgL_3 & \bgL_4} \in \LIT_{N,\Fq},
    \]
    where the matrices $\bgL_i \in \diagonal_{N,\Fq},\ i \in [4]$, are such that $\det(\bgL_1 \bgL_4 - \bgL_2 \bgL_3) \neq 0$.
    Let $\bgL_d = \bgL_1 \bgL_4 - \bgL_2 \bgL_3$ and $\bgD = \bgL_d^{-1} \in \diagonal_{N,\Fq}$, then
    \[ \begin{split}
        \bgL \ppmatrix{\bI & \bzero \\ \bzero & \bgD} \bJ \ppmatrix{\bI & \bzero \\ \bzero & \bgD}^\top \bgL^\top & = 
        \bgL \ppmatrix{\bzero & -\bI \\ \bgD & \bzero} \ppmatrix{\bI & \bzero \\ \bzero & \bgD} \bgL^\top  = \bgL \ppmatrix{\bzero & -\bgD \\ \bgD & \bzero} \bgL^\top \\
        & = \ppmatrix{\bgL_2 \bgD & -\bgL_1 \bgD \\ \bgL_4 \bgD & -\bgL_3 \bgD} \ppmatrix{\bgL_1 & \bgL_3 \\ \bgL_2 & \bgL_4} \\
        & = \ppmatrix{\bgL_2 \bgD \bgL_1 - \bgL_1 \bgD \bgL_2 & \bgL_2 \bgD \bgL_3 - \bgL_1 \bgD \bgL_4 \\ \bgL_4 \bgD \bgL_1 - \bgL_3 \bgD \bgL_2 & \bgL_4 \bgD \bgL_3 - \bgL_3 \bgD \bgL_4} \\
        & = \ppmatrix{(\bgL_2 \bgL_1 - \bgL_1 \bgL_2) \bgD & (\bgL_2 \bgL_3 - \bgL_1 \bgL_4) \bgD \\ (\bgL_4 \bgL_1 - \bgL_3 \bgL_2) \bgD & (\bgL_4 \bgL_3 - \bgL_3 \bgL_4) \bgD} \\
        & = \ppmatrix{(\bgL_1 \bgL_2 - \bgL_1 \bgL_2) \bgD & - (\bgL_1 \bgL_4 - \bgL_2 \bgL_3) \bgD \\ (\bgL_1 \bgL_4 - \bgL_2 \bgL_3) \bgD & (\bgL_3 \bgL_4 - \bgL_3 \bgL_4) \bgD} = \bJ
    \end{split} \]
    by the commutativity of $\diagonal_{N,\Fq}$.
    It follows that
    \[ \begin{split}
        (\bM')^\top \bJ \bM' & = \bM^\top \ppmatrix{\bI & \bzero \\ \bzero & \bgD} \bJ \ppmatrix{\bI & \bzero \\ \bzero & \bgD}^\top \bM \\
        & = \bP \ppmatrix{\bI & \bS} \bgL \ppmatrix{\bI & \bzero \\ \bzero & \bgD} \bJ \ppmatrix{\bI & \bzero \\ \bzero & \bgD}^\top \bgL^\top \ppmatrix{\bI & \bS}^\top \bP^\top \\
        & = \bP \ppmatrix{\bI & \bS} \bJ \ppmatrix{\bI & \bS}^\top \bP^\top = \bzero,
    \end{split} \]
    \ie $\bM' \in \cM_o$.

    For the other direction it is trivial to see that if such a $\bgD$ exists, then $\bM \liteq \bM' \in \cM_o \subseteq \cM_{\LIT}$, which implies that $\bM \in \cM_{\LIT}$.
\end{proof}

\begin{remark}
    Notice that in the proof we proved that the matrix $\bgL' = \bgL \ppsmatrix{\bI & \bzero \\ \bzero & \bgD}$ is symplectic, \ie it satisfies the property $\bgL' \bJ (\bgL')^\top = \bJ$, as symplectic matrices are known to be the only kind of matrices that preserves strong self-orthogonality.
\end{remark}


\section{$N$-Sum Box Application: Quantum CSA Scheme}
\label{sec:QCSA}

Cross-subspace alignment (CSA) codes find applications in various private information retrieval (PIR) schemes (\eg PIR with secure storage) and in secure distributed batch matrix multiplication (SDBMM).
Using the developed $N$-sum box abstraction of a quantum multiple-access channel (QMAC), we translate CSA schemes over classical multiple-access channels into efficient quantum CSA schemes over a QMAC, achieving maximal superdense-coding gain. 
Because of the $N$-sum box abstraction, the underlying problem of coding to exploit quantum entanglements for CSA schemes becomes conceptually equivalent to that of designing a channel matrix for a MIMO MAC subject to given structural constraints imposed by the $N$-sum box abstraction, such that the resulting MIMO MAC is able to implement the functionality of a CSA scheme (encoding/decoding) \emph{over-the-air}. 
Applications include Quantum PIR with secure and MDS-coded storage, as well as Quantum SDBMM. 
In this section we first introduce the classical CSA scheme, then give the definition of some important concepts and finally present the way to translate a CSA scheme to QCSA scheme and apply the QCSA scheme to specific problems.

\subsection{Cross Subspace Alignment (CSA) Codes}

Conceptually, the setting for a CSA coding scheme (over a finite field $\Fq$) is the following.  
We have $N$ distributed servers. 
The servers locally compute their answers $A_n, n\in[N],$ to a user's query. 
Each answer $A_n$ is a linear combination  of $L$ symbols that are desired by the user, say $\delta_1, \delta_2, \cdots, \delta_L$, and $N-L$ symbols of undesired information (interference), say $\nu_1, \nu_2, \cdots, \nu_{N-L}$. 
The linear combinations have a Cauchy-Vandermonde structure that is the defining characteristic of CSA schemes, such that the desired terms appear along the dimensions corresponding to the Cauchy terms, while the interference  appears (aligned) along the dimensions corresponding to the Vandermonde terms. 
The $b^{th}$ instance of a CSA scheme is represented as
\begin{align*}
    \underbrace{\ppmatrix{
    A^b_1\\ \vdots \\ A^b_N
    }}_{\bA^b}
    =
    \underbrace{\ppmatrix{
    \begin{array}{ccc|cccc}
    \frac{1}{f_1-\alpha_1}&\cdots&\frac{1}{f_L-\alpha_1}&1&\alpha_1&\cdots&\alpha_1^{N-L-1}\\
    \frac{1}{f_1-\alpha_2}&\cdots&\frac{1}{f_L-\alpha_2}&1&\alpha_2&\cdots&\alpha_2^{N-L-1}\\
    \vdots&\vdots&\vdots&\vdots&\vdots&\vdots&\vdots\\
    \frac{1}{f_1-\alpha_N}&\cdots&\frac{1}{f_L-\alpha_N}&1&\alpha_N&\cdots&\alpha_N^{N-L-1}\\
    \end{array}
    }}_{\bG_{\tiny \mathrm{CSA}^q_{N,L}(\bm{\alpha},\bff)}}
    \underbrace{\ppmatrix{
    \delta^b_1\\
    :\\
    \delta^b_L\\
    \nu^b_1\\
    :\\
    \nu^b_{\mbox{\tiny $N-L$}}
    }}_{\bX^b_{\delta,\nu}}.
\end{align*}
The CSA scheme requires that all the $\alpha_i, f_j$ are distinct elements in $\Fq$ (thus needing $q \geq N+L$), which guarantees that the $N \times N$ matrix $\bG_{\tiny \mathrm{CSA}^q_{N,L}(\bm{\alpha},\bff)}$ is invertible. 
After downloading $A_n$ from each server $n, n\in[N]$, the user is able to recover the desired symbols $\delta^b_1, \ldots, \delta^b_L$ by inverting $\mathrm{CSA}^q_{N,L}(\bm{\alpha},\bff)$. 
Thus, each instance $b$ of the CSA scheme allows the user to retrieve $L$ desired symbols at a cost of $N$ downloaded symbols. 
The rate of the scheme, defined as the number of desired symbols recovered per downloaded symbol, is $L/N$. 
The reciprocal, $N/L$, is the download cost per desired symbol.

\begin{remark} \label{rem:redundant} 
    A noteworthy aspect of CSA schemes is that the number of servers can be reduced, \ie the CSA scheme can be applied to $N'<N$ servers, with a corresponding reduction in the number of desired symbols $L'<L$, as long as the dimension of interference is preserved, \ie $N-L=N'-L'$. 
    In the classical setting, this flexibility is not useful as it leads to a strictly higher download cost, \ie $N'/L' = (N-L)/L' + 1>(N-L)/L+1=N/L$. 
    In the quantum setting, however, this will lead to a useful simplification.
\end{remark}

\subsection{Definitions}

In the following we represent GRS codewords as column vectors rather than the usual coding-theoretic notation with row vectors in order to be consistent with the standard notation used in linear computation.
It follows that the GRS generator matrix is represented as a $n \times k$ matrix instead of the usual $k \times n$ matrix.

\begin{definition}{\bf (GRS Code)} 
    Let $\cC = \mathrm{GRS}^q_{n,k}(\bm{\alpha},\bu)$ be a Generalized Reed--Solomon (GRS) code over $\Fq$, where $\bm{\alpha} = (\alpha_1,\ldots,\alpha_n) \in \Fq^{1\times n}$, $\bu = (u_1,\ldots,u_n) \in \Fq^{1\times n}$, $u_i\neq 0$ and $\alpha_i\neq \alpha_j$ for all distinct $i,j\in[n]$.
    Its generator matrix can be defined as
    \begin{equation*}
        \bG_{\mathrm{GRS}^q_{n,k}(\bm{\alpha},\bu)} \coloneqq \ppmatrix{
        u_1&u_1\alpha_1&u_1\alpha_1^2&\cdots&u_1\alpha_1^{k-1}\\[0.1cm]
        u_2&u_2\alpha_2&u_2\alpha_2^2&\cdots&u_2\alpha_2^{k-1}\\[0.1cm]
        \vdots&\vdots&\vdots&\vdots&\vdots\\[0.1cm]
        u_n&u_n\alpha_n&u_n\alpha_n^2&\cdots&u_n\alpha_n^{k-1}\\[0.1cm]
        } \in \Fq^{n \times k}.
    \end{equation*}
\end{definition}

\begin{definition}{\bf (Dual Code)}
    For an $[n,k]$ linear code $\cC\subset \Fq^{n\times 1}$, its dual code $\cC^\perp$ is defined as $$\cC^\perp \coloneqq \px*{\bv \in \Fq^{n \times 1}: \langle\bv,\bc\rangle=0, \forall \bc\in \cC }.$$
\end{definition}

The dual code of a GRS code is also a GRS code, \eg according to the following construction \cite{Macwilliams}.
\begin{definition}{{\bf (Dual GRS Code)}\cite{Macwilliams}}
    For $\bv=(v_1,\ldots,v_n)\in\Fq^{1\times n}$ defined as
    \begin{equation} \label{eq:def_v}
        v_j =\frac{1}{u_j} \paren*{\prod_{i\in[n],i\neq j }(\alpha_j-\alpha_i)}^{-1} \quad \forall j\in[n],
    \end{equation}
    where $u_j \neq 0$, we have
    \begin{equation}
        \bG_{\mathrm{GRS}^q_{n,k}(\bm{\alpha},\bu)}^\top \bG_{\mathrm{GRS}^q_{n,n-k}(\bm{\alpha},\bv)} = \bzero_{k\times (n-k)},\label{eq:guarantee}
    \end{equation}
    \ie $\mathrm{GRS}^q_{n,n-k}(\bm{\alpha},\bv)$ is an $[n,n-k]$ code that is the dual of the code $\mathrm{GRS}^q_{n,k}(\bm{\alpha},\bu)$.
\end{definition}

\begin{definition}{\bf (QCSA Matrix)} \label{def:QCSA} 
    We define the $\mathrm{QCSA}^q_{N,L}(\bm{\alpha},\bm{\beta},\bff)$ matrix as the $N \times N$ matrix
    \begin{align}
        &\bG_{\tiny \mathrm{QCSA}^q_{N,L}(\bm{\alpha},\bm{\beta},\bff)}
        \coloneqq \notag\\[0.1cm]
        &\scalebox{.85}{$
        \ppmatrix{\begin{array}{cccc|ccccccccc}
        \frac{\beta_1}{f_1 - \alpha_1} & \frac{\beta_1}{f_2 - \alpha_1} & \cdots & \frac{\beta_1}{f_L - \alpha_1}
        & \overmat{\bG_{\mathrm{GRS}^q_{N,\lfloor N/2\rfloor}}(\bm{\alpha},\bm{\beta})}{\beta_1 & \beta_1 \alpha_{1} & \cdots &  \beta_1 \alpha_{1}^{\lfloor N/2\rfloor - 1}} & \beta_1 \alpha_{1}^{\lceil N/2\rceil - 1}
        &\beta_1 \alpha_{1}^{\lceil N/2\rceil} & \cdots & \beta_1 \alpha_{1}^{N-L-1}\\[0.1cm]
        \frac{\beta_2}{f_1 - \alpha_2} & \frac{\beta_2}{f_2 - \alpha_2} & \cdots & \frac{\beta_2}{f_L - \alpha_2}
        &\beta_2 & \beta_2 \alpha_{2} & \cdots & \beta_{2} \alpha_{2}^{\lfloor N/2\rfloor - 1} & \beta_2 \alpha_{2}^{\lceil N/2\rceil - 1}
        &\beta_2 \alpha_{2}^{\lceil N/2\rceil} & \cdots & \beta_2 \alpha_{2}^{N-L-1}\\
        \vdots & \vdots & \vdots & \vdots&\vdots & \vdots & \vdots&\vdots & \vdots & \vdots\\
        \undermat{\tilde{\bH}^{\bm{\beta}}_\sC}{\frac{\beta_N}{f_1 - \alpha_N} & \frac{\beta_N}{f_2 - \alpha_N} & \cdots & \frac{\beta_N}{f_L - \alpha_N}}&
        \undermat{\bG_{\mathrm{GRS}^q_{N,\lceil N/2\rceil}(\bm{\alpha},\bm{\beta})}} {\beta_N & \beta_N \alpha_{N} & \cdots & \beta_N \alpha_{N}^{\lfloor N/2\rfloor - 1} & \beta_N \alpha_{N}^{\lceil N/2\rceil - 1}}&
        \undermat{\tilde{\bH}^{\bm{\beta}}_\sV} {\beta_N \alpha_{N}^{\lceil N/2\rceil} & \cdots & \beta_N \alpha_{N}^{N-L-1}}
        \end{array}}$
        }, \label{eq:MCV_Blocks}\\
        \notag
    \end{align} 
    where $\bm{\alpha}=(\alpha_1,\ldots,\alpha_N)\in\Fq^{1\times N}$, $\bm{\beta}=(\beta_1,\ldots,\beta_N)\in\Fq^{1\times N}$, $\bff=(f_1,f_2,\ldots,f_L)\in\Fq^{1\times L}$, all $\beta_i, i\in[N]$ are non-zero, the terms $\alpha_1,\ldots, \alpha_N,f_1,\ldots,f_L$ are $L+N$ distinct elements of $\Fq$, $q\geq L+N$, $L\leq N/2$.  
    The sub-matrices $\bG_{\mathrm{GRS}^q_{N,\lfloor N/2\rfloor}(\bm{\alpha},\bm{\beta})}$ and $\bG_{\mathrm{GRS}^q_{N,\lceil N/2\rceil}(\bm{\alpha},\bm{\beta})}$ identified in Equation~\eqref{eq:MCV_Blocks} are the generator matrices of $[N, \floor*{N/2}]$ and $[N, \ceil*{N/2}]$ GRS codes, respectively.
    As a special case, setting $\beta_1=\ldots=\beta_N=1$, we have,
    \begin{align}
        \bG_{\mathrm{CSA}^q_{N,L}(\bm{\alpha},\bff)} = \bG_{\mathrm{QCSA}^q_{N,L}(\bm{\alpha},\mathbf{1},\bff)}.
    \end{align}
\end{definition}

\begin{remark}
    When $N$ is an even number, the two sub-matrices $\bG_{\mathrm{GRS}^q_{N,\lfloor N/2\rfloor}(\bm{\alpha},\bm{\beta})}$ and $\bG_{\mathrm{GRS}^q_{N,\lceil N/2\rceil}(\bm{\alpha},\bm{\beta})}$ become the same sub-matrix $\bG_{\mathrm{GRS}^q_{N, N/2}(\bm{\alpha},\bm{\beta})}$.
\end{remark}

\begin{remark} 
The $\mathrm{QCSA}_{N,L}^{q}(\bm{\alpha},\bm{\beta},\bff)$ matrix is the product of an $N\times N$ diagonal matrix $\diagonal(\beta_1,\beta_2,\cdots,\beta_N)$ and the $N\times N$ Cauchy-Vandermonde matrix $\bG_{\mathrm{CSA}^q_{N,L}(\bm{\alpha},\bff)}$ invoked by CSA schemes (cf. \cite[Equation (11)]{Jia_Jafar_MDSXSTPIR}). 
The former is full rank since $\beta_n\neq 0,\ \forall n\in [N]$, and the latter is full rank according to \cite[Lemma 1]{Jia_Jafar_MDSXSTPIR} because $\alpha_1,\cdots, \alpha_N,f_1,\cdots, f_N$ are distinct ($q\geq L+N$). 
Since multiplication with an invertible matrix preserves rank, the $\mathrm{QCSA}^q_{N,L}(\bm{\alpha},\bm{\beta},\bff)$ matrix is an invertible matrix. 
\end{remark}

\subsection{From CSA Scheme to QCSA Scheme}\label{subsec:CSA2QCSA}

Given a  CSA scheme over $\Fq$, we show how to translate it into a QCSA scheme over a QMAC. 
The process is described by the following three steps.
\begin{enumerate}
    \item From the $\mathrm{CSA}_{N,L}^{q}(\bm{\alpha}, \bff)$ matrix, construct two matrices $\bQ_N^{\bu} = \bG_{\mathrm{QCSA}_{N,L}^{q}(\bm{\alpha}, \bu, \bff)},\ \bQ_N^{\bv} = \bG_{\mathrm{QCSA}_{N,L}^{q}(\bm{\alpha}, \bv, \bff)}$.
    
    \item From $\bQ_N^{\bu}, \bQ_N^{\bv}$, construct a feasible $N$-sum Box, \ie a MIMO MAC with channel matrix $\bM_{\mathrm{QCSA}}$.
    
    \item Over the MIMO MAC, realize `over-the-air' decoding of two instances of the CSA scheme. 
    By the $N$-sum box abstraction, this automatically maps to a quantum protocol (a QCSA scheme) and the efficiency gained by `over-the-air' decoding in the MIMO MAC translates into the superdense coding gain over the QMAC.
\end{enumerate}
These steps are explained next.

\subsubsection*{Step 1. Generation of QCSA matrices}
Let $\bu = (u_1, \ldots, u_N) \in \Fq^{1\times N}$ and $\bv = (v_1, \ldots, v_N) \in \Fq^{1\times N}$, where
\begin{equation} \label{eq:b_compute}
    v_j = \frac{1}{u_j} \paren*{\prod_{i\in[N],i\neq j }(\alpha_j-\alpha_i)}^{-1} \quad \forall j\in[N],
\end{equation}
such that the sub-matrix $\bG_{\mathrm{GRS}^q_{N,\ceil*{N/2}}(\bm{\alpha},\bu)}$ of $\bQ_N^{\bu}$ and the sub-matrix $\bG_{\mathrm{GRS}^q_{N,\floor*{N/2}}(\bm{\alpha},\bv)}$ of $\bQ_N^{\bv}$ satisfy
\begin{equation} \label{eq:dualsub}
    \bG_{\mathrm{GRS}^q_{N,\ceil*{N/2}}(\bm{\alpha},\bu)}^\top \bG_{\mathrm{GRS}^q_{N,\floor*{N/2}}(\bm{\alpha},\bv)} = \bzero_{\ceil*{N/2} \times \floor*{N/2}},
\end{equation}
according to Equation~\eqref{eq:def_v} and Equation~\eqref{eq:guarantee}. 

\subsubsection*{Step 2. A suitable $N$-sum box}
The $N$-sum box is specified by the following theorem.
\begin{theorem} \label{thm:QCSA}
    For the $\bQ_N^{\bu}$ and $\bQ_N^{\bv}$ constructed in Step 1, there exists a feasible $N$-sum box $\by=\bM_{\mathrm{QCSA}}\bx$ in $\Fq$ with the $N\times 2N$ transfer matrix,
    \begin{equation} \label{eq:thmQCSA}
        \bM_{\mathrm{QCSA}} = \ppmatrix{\begin{array}{ccc|ccc}
        \bI_L&\bzero_{L\times \lceil N/2\rceil}&\bzero&\bzero&\bzero&\bzero\\
        \bzero&\bzero&\bI_{\lfloor N/2\rfloor -L}&\bzero&\bzero&\bzero\\
        \bzero&\bzero&\bzero&\bI_L&\bzero_{L\times \lfloor N/2\rfloor }&\bzero\\
        \bzero&\bzero&\bzero&\bzero&\bzero&\bI_{\lceil N/2\rceil -L}
        \end{array}}
        \ppmatrix{
        \bQ_N^{\bu}&\bzero\\
        \bzero&\bQ_N^{\bv}
        }^{-1}.
    \end{equation}
\end{theorem}
A proof of Theorem~\ref{thm:QCSA} is presented in Appendix~\ref{app:proofQCSA}. 

\subsubsection*{Step 3. QCSA Scheme as `Over-the-Air' CSA}\label{sec:csa2qcsa}
Using the MIMO MAC with channel matrix $\bM_{\mathrm{QCSA}}$ identified in Step 2, we now describe how to achieve `over-the-air' decoding of several instances of CSA schemes. 
Since the MIMO MAC is actually an $N$-sum box, which in fact represents a quantum protocol with communication cost $N$ qudits, a QCSA scheme is automatically implied for the QMAC through the $N$-sum box abstraction. 

First consider the case where we are given a CSA scheme with $L\leq N/2$. 
With 2 instances of the CSA scheme we have
\begin{equation}
    \ppmatrix{\bA^1 \\ \bA^2} =
    \ppmatrix{
    \bG_{\mathrm{CSA}^q_{N,L}(\bm{\alpha},\bff)} &\bzero \\
    \bzero & \bG_{\mathrm{CSA}^q_{N,L}(\bm{\alpha},\bff)}
    }
    \ppmatrix{
    \bX_{\delta,\nu}^1 \\ \bX_{\delta,\nu}^2},
\end{equation}
which retrieves $2L$ desired symbols at the download cost of $2N$ symbols. 
Now the corresponding QCSA scheme (over-the-air MIMO MAC) is obtained as follows.
\begin{align}
    \by & = \bM_{\mathrm{QCSA}}
    \underbrace{\ppmatrix{
    \diagonal(\bu) & \bzero\\
    \bzero & \diagonal(\bv)
    }
    \ppmatrix{
    \bA^1\\
    \bA^2
    }}_{\bx} \label{eq:serverscale} \\
    & = \bM_{\mathrm{QCSA}} \ppmatrix{
    \diagonal(\bu) & \bzero \\
    \bzero & \diagonal(\bv)
    }
    \ppmatrix{
    \bG_{\mathrm{CSA}^q_{N,L}(\bm{\alpha},\bff)} & \bzero \\
    \bzero & \bG_{\mathrm{CSA}^q_{N,L}(\bm{\alpha},\bff)}
    }
    \ppmatrix{
    \bX^1_{\delta,\nu}\\
    \bX^2_{\delta,\nu}
    } \\
    & = \bM_{\mathrm{QCSA}} \ppmatrix{
    \bQ_N^{\bu} & \bzero\\
    \bzero & \bQ_N^{\bv}
    } \ppmatrix{
    \bX^1_{\delta,\nu}\\
    \bX^2_{\delta,\nu}
    } \\
    &=
    \ppmatrix{\begin{array}{ccc|ccc}
    \bI_L&\bzero_{L\times \lceil N/2\rceil}&\bzero&\bzero&\bzero&\bzero\\
    \bzero&\bzero&\bI_{\lfloor N/2\rfloor -L}&\bzero&\bzero&\bzero\\
    \bzero&\bzero&\bzero&\bI_L&\bzero_{L\times \lfloor N/2\rfloor }&\bzero\\
    \bzero&\bzero&\bzero&\bzero&\bzero&\bI_{\lceil N/2\rceil -L}
    \end{array}}
    \ppmatrix{
    \bX^1_{\delta,\nu} \\
    \bX^2_{\delta,\nu}
    }\notag\\
    & = \paren*{\delta^1_1,\cdots,\delta^1_L,\bm{\nu}^1_{(\leftharpoondown)}, \delta^2_1,\cdots,\delta^2_L,\bm{\nu}^2_{(\leftharpoondown)}}^\top \label{eq:qcsa_output}.
\end{align}
where the $N$ entries of $\bu$ are non-zero, $\bv$ is specified in Equation~\eqref{eq:b_compute}, $\bm{\nu}^1_{(\leftharpoondown)}$ represents the last $\lfloor N/2\rfloor -L$ symbols of the vector $\bm{\nu}^1 = \paren*{\nu^1_1,\ldots,\nu^1_{N-L}}$, and $\bm{\nu}_{(\leftharpoondown)}(2)$ represents the last $\lceil N/2\rceil -L$ symbols of the vector $\bm{\nu}^2 = \paren*{\nu^2_1,\ldots,\nu^2_{N-L}}$.  
Multiplication by $\diagonal(\bu,\bv)$ in Equation~\eqref{eq:serverscale} simply involves each server $j\in[N]$ scaling its answers of the two instances of the CSA scheme $A^1_j, A^2_j$ by $u_j, v_j$, respectively, and applying $u_j A^1_j, v_j A^2_j$ to the inputs of the $N$-sum box (MIMO MAC) corresponding to that server. 
Evidently, all $2L$ desired symbols are recovered, and the total download cost is $N$ qudits (one qudit from each server), for a normalized download cost of $N/(2L)$ qudits per desired dit. 
The improvement from $N/L$ (classical CSA) to $N/2L$ (QCSA) reflects the factor of $2$ superdense coding gain in communication efficiency.

If $L > N/2$, we discard `redundant' servers (cf. Remark~\ref{rem:redundant}) and only employ $N' = 2N - 2L < N$ servers, choosing a CSA scheme with $L' =  N' / 2$, such that $N - L = N' - L'$, \ie the dimensions of interference are preserved, which results in a download cost of $N'/(2L') = 1$ qudit per desired symbol.

\subsection{QCSA Scheme Application}

Based on the approach described above, existing achievability results based on CSA schemes in the classical setting translate into corresponding achievability results for QCSA schemes in the quantum setting. 
In particular, the following corollaries follow immediately from this approach.
\begin{corollary} \label{cor:mdsxstpir}
    For the $(N,M,K,X,T)$ MDSXSTPIR \cite{Jia_Jafar_MDSXSTPIR}, where $N>X+T+K-1$, a QCSA scheme achieves rate
    \begin{equation}
        R_\sQ = \min\px*{1, 2 \paren*{1 - \paren*{\frac{X+T+K - 1}{N}}}},
    \end{equation}
    \ie the user is able to recover $L = N R_\sQ$ $q$-ary symbols of desired information for every $N$ $q$-dimensional qudits that it downloads from the servers.
\end{corollary}
\begin{remark}
    In an MDSXSTPIR problem, there are $M$ messages and $N$ servers. Every $K$ symbols of one message, together with $X$ random noise symbols, are encoded according to an $[N,K + X]$ MDS code and each codeword symbol is stored at one of the $N$ servers. Thus, the storage cost of every server is $\frac{1}{M}$ of the size of all the $M$ messages, and any group of up to $X$ colluding servers learn nothing about the $M$ messages. A user wishes to retrieve one of the $M$ messages by querying the $N$ distributed servers such that any group of up to $T$ colluding servers can learn nothing about  the desired message index.
\end{remark}
Appendix~\ref{app:proof_mdsxstpir} provides a proof of Corollary~\ref{cor:mdsxstpir}. 
In addition to establishing achievable rates for quantum versions of XSTPIR \cite{Jia_Sun_Jafar_XSTPIR}, MDS-XSTPIR \cite{Jia_Jafar_MDSXSTPIR} that have not been previously explored, Corollary~\ref{cor:mdsxstpir} recovers existing achievability results in Quantum PIR for the cases of TPIR \cite{QTPIR} and MDS-TPIR \cite{QMDSTPIR}. 
Relative to prior works, the field size required by the QCSA scheme is linear in $N$, because an even $q \geq L+N$ where $L \leq \frac{N}{2}$ guarantees the existence of the $\mathrm{QCSA}$ matrix, while in \cite{QTPIR}, the field size is exponential in $N$, due to an $\cO(N)$ fold field extension step. 
It is also noteworthy that prior quantum PIR schemes employ mixed quantum states to achieve \emph{symmetric} privacy, \ie the user does not learn more than his desired information. 
While the QCSA scheme does not automatically ensure symmetric privacy, the same can be accomplished by noise alignment based on shared common randomness among servers as in \cite{Chen_Jia_Wang_Jafar_NGCSA}. 
In Section~\ref{sec:applications_kn_sum_box} we discuss how we can achieve symmetric privacy automatically with a non-maximal stabilizer based construction.
Shared (classical) common randomness among servers is not difficult to achieve when the servers share quantum entanglements.

\begin{corollary}\label{cor:sdmm}
    For the $(N,X_A,X_B), N > X_A + X_B$ SDBMM (secure distributed batch matrix multiplication) problem defined in \cite{Jia_Jafar_SDMM}, where $L$ matrices $\bA_1, \ldots, \bA_{L} \in \Fq^{\lambda\times \eta}$ are $X_A$-securely shared among $N$ servers, another $L$ matrices $\bB_1, \ldots, \bB_{L} \in \Fq^{\eta\times \mu}$ are $X_B$-securely shared among the same $N$ servers, and the user wants to compute the $L$ products $\bA_1 \bB_1, \bA_2 \bB_2, \ldots, \bA_{L} \bB_{L} \in \Fq^{\lambda\times \mu}$ by querying the $N>X_A+X_B$ servers, a QCSA scheme achieves the rate
    \begin{align}
        R_\sQ = \min\px*{1, 2 \paren*{1 - \paren*{\frac{X_A + X_B}{N}}}},
    \end{align}
    \ie from every $N \cdot (\lambda\mu)$ $q$-dimensional qudits downloaded from $N$ servers, the user recovers $L= NR_\sQ$ desired product matrices in $\Fq^{\lambda\times \mu}$.
\end{corollary}
Section~\ref{proof:sdmm} provides a proof of Corollary~\ref{cor:sdmm}. Let us note that in the original problem defined in \cite{Jia_Jafar_SDMM}, the $S$ matrices $\bA_1, \ldots, \bA_S \in \Fq^{L\times K}$ and another $S$ matrices $\bB_1, \ldots, \bB_S \in \Fq^{K\times M}$ are to be pairwise multiplied. 
The parameters $(S,L,K,M)$ from \cite{Jia_Jafar_SDMM} are mapped to $(L,\lambda,\eta,\mu)$ to fit the notation in this work.

\begin{remark}
    In both corollaries, the rate achieved with the QCSA scheme can be expressed as $R_\sQ = \min\{1,2 R_\sC\}$, where $R_\sC$ is the rate achieved by the CSA scheme in the corresponding classical setting. 
    Considering that rates greater than 1 are not possible according to the Holevo bound \cite{Holevo_Bounds}, it is apparent that the QCSA scheme  achieves the maximal superdense-coding gain relative to the classical CSA scheme.
\end{remark}


\section{Non-Maximal-Stabilizer-Based Constructions}
\label{sec:non_max_stab_based_constructions}

In Section~\ref{sec:stab_formalism} we noticed that the output of the measurement associated with a non-maximal stabilizer has less digits than one would expect.
The reason is that the states that form a basis for the space $\cW$ (cf. Equation~\eqref{eq:Hdecomposition}) associated to a non-maximal stabilizer are mixed states, so the remaining $N-\kappa$ output digits result as uniformly random and can be \emph{discarded}.
Thus, non-maximal stabilizers can be useful to describe black boxes that output only $\kappa$ $q$-ary digits while discarding the last $N-\kappa$ digits, which we call \emph{$(\kappa,N)$-sum boxes}.

Furthermore, given a maximal stabilizer with $N$ generators $\cS(\cV')$, one can always obtain a non-maximal stabilizer $\cS(\cV)$ by choosing $\kappa$ of those.
Conversely, given a non-maximal stabilizer with $\kappa$ generators $\cS(\cV)$, one can always complete the generator basis to obtain a maximal stabilizer $\cS(\cV')$.

These two observations suggest that, in order to obtain the stabilized state of a non-maximal stabilizer $\cS(\cV)$, one can start with $N$ qudits in the mixed state $\ket{\bzero} \bra{\bzero} \otimes \frac{\bI_{q^{N-\kappa}}}{q^{N-\kappa}}$ and apply the same unitary as one would use to generate $\ket{\bzero}_\cW$ from $\ket{\bzero}$ in a maximal stabilizer $\cS(\cV')$ that completes $\cS(\cV)$.
After we apply a unitary $\tilde{\bW}(\bs)$ for some $\bs \in \Fq^{2N}$ and measure using $\cP^{\cV'}$, the output will have $N-\kappa$ uniformly random digits.
\begin{example}
    \label{ex:non_max_stab}
    Let $\bG = \ppmatrix{1 & 1 & 0 & 0}^\top$ and $\cV = \colspan{\bG}$.
    Clearly $\bG$ is a submatrix of the matrix determined by $\cS$ in Example~\ref{ex:two_sum_box}, so we can take $\cS$ as the maximal-stabilizer completion of $\cS(\cV)$.
    Consider $\bs = (x_1,x_2,x_3,x_4) \in \F_2^4$. 
    Let $\rho = \ket{0} \bra{0} \otimes \frac{\bI}{2}$ be the original state, and let $\rho' = \ket{\varphi} \bra{\varphi}$ be its purification, where $\ket{\varphi} = \ket{0} \otimes (\ket{00} + \ket{11})/\sqrt{2}$.
    Here we consider a purification to clarify the evolution of the state before the measurement.
    By applying the same gates as one would use to prepare the initial entangled state for the two-sum transmission protocol, \ie a Hadamard gate on the first qubit and a CNOT gate with the first qubit as control and the second qubit as target (cf. Example~\ref{ex:two_sum_box}), we obtain the initial pure state
    \[
    \ket{\varphi'} = (\ket{000} + \ket{110} + \ket{101} + \ket{011}) / 2,
    \]
    which is stabilized by the three independent Weyl operators $\bW(1,1,0,0,0,0),\ \bW(1,0,1,0,0,0)$ and $\bW(0,1,1,0,0,0)$.

    The state evolves in the following way:
    \[\begin{split}
        \ket{\varphi'} \stackrel{\footnotesize \sZ(x_3) \otimes \sZ(x_4) \otimes \bI}{\rightarrow} & \paren*{\ket{000} + (-1)^{x_3+x_4} \ket{110} + (-1)^{x_3} \ket{101} + (-1)^{x_4} \ket{011}} / 2 \\
        \stackrel{\footnotesize \sX(x_1) \otimes \sX(x_2) \otimes \bI}{\rightarrow} & \paren*{\ket{x_1 x_2 0} + (-1)^{x_3+x_4} \ket{\overline{x_1} \overline{x_2} 0} + (-1)^{x_3} \ket{\overline{x_1} x_2 1} + (-1)^{x_4} \ket{x_1 \overline{x_2} 1}} / 2 \\
        \stackrel{\footnotesize \mathsf{CNOT}_{1,2} \otimes \bI}{\rightarrow} 
        & \bigl(\ket{x_1 (x_1 + x_2) 0} + (-1)^{x_3+x_4} \ket{\overline{x_1} (x_1 + x_2) 0} \\
        & + (-1)^{x_3} \ket{\overline{x_1} (\overline{x_1 + x_2}) 1} + (-1)^{x_4} \ket{x_1 \overline{x_1 + x_2} 1} \bigr) / 2 = \\
        & \paren*{\paren*{\ket{x_1} + (-1)^{x_3 + x_4} \ket{\overline{x_1}}} / \sqrt{2}} \otimes \paren*{\paren*{\ket{(x_1 + x_2) 0} + (-1)^{x_4} \ket{(\overline{x_1+x_2} 1)}} / \sqrt{2}} = \\
        & \paren*{\sX(x_1) \paren*{\ket{0} + (-1)^{x_3 + x_4} \ket{1}} / \sqrt{2}} \otimes \paren*{(\sX(x_1+x_2) \sZ(x_4) \otimes \bI) \paren*{\ket{00} + \ket{11}} / \sqrt{2} } \\
        \stackrel{\footnotesize \sH \otimes \bI \otimes \bI}{\rightarrow} 
        & \paren*{\sZ(x_1) \ket{x_3 + x_4}} \otimes \paren*{(\sX(x_1+x_2) \sZ(x_4) \otimes \bI) (\ket{00} + \ket{11}) / \sqrt{2}}
    \end{split}\]
    Thus, after applying the Weyl operator $\bW(\bs)$ and reverting the initial CNOT and Hadamard gates we end up with the qubits in the state
    \[
    \ket{\varphi''} = (-1)^{x_1+x_3+x_4} \ket{x_3+x_4} \otimes \paren*{\bW(x_1+x_2,0,x_4,0) (\ket{00} + \ket{11})/\sqrt{2}}.
    \]
    Let $\ttr : \cH \to \R$ be the trace operator that maps a state to the trace of its density matrix. Let $\ttr: \cH_1 \otimes \cH_2 \otimes \cH_3 \to \cH_1 \otimes \cH_2$ the partial trace over the third quantum system. In this case, given the density matrix $\bD$, the partial trace is given by
    \[
    \ttr(\bD) = \sum_{x=0}^1 \paren*{\bI \otimes \bI \otimes \bra{x}} \bD \paren*{\bI \otimes \bI \otimes \ket{x}}.
    \]
    Notice that, by the properties of the tensor product, the state over the first qubit has density matrix $\ket{x_3+x_4}\bra{x_3+x_4}$ after applying the partial trace operator. On the other side, 
    \[\begin{split}
        & \frac{1}{2} \sum_{x=0}^1 (\bI \otimes \bra{x}) \paren*{\bW(x_1+x_2,0,x_4,0) (\ket{00} + \ket{11}) (\bra{00} + \bra{11}) \bW^\dagger(x_1+x_2,0,x_4,0)} (\bI \otimes \ket{x}) = \\
        & \frac{1}{2} \bW(x_1 + x_2, x_4) \paren*{\sum_{x=0}^1 (\bI \otimes \bra{x}) (\ket{00} + \ket{11}) (\bra{00} + \bra{11}) (\bI \otimes \ket{x})} \bW^\dagger(x_1 + x_2, x_4) = \\
        & \frac{1}{2} \bW(x_1 + x_2, x_4) \bI \bW^\dagger(x_1 + x_2, x_4) = \frac{\bI}{2}.
    \end{split}\]
    We conclude that, by tracing out the third qubit, the state has the density matrix
    \[
    \rho_{out} = \ket{x_3+x_4} \bra{x_3+x_4} \otimes \frac{\bI}{2},
    \]
    which outputs the bit $x_3 + x_4$ and a random bit.
\end{example}

\subsection{Stabilizer-Based $(\kappa,N)$-Sum Boxes}
\label{sec:stab_based_kn_sum_boxes}


The setting for $(\kappa,N)$-sum boxes is similar to the one described for $N$-sum boxes in Section~\ref{sec:n_sum_box}, with the difference that the transfer matrix $\bM$ has dimensions $\kappa \times 2N$ and $N-\kappa$ outputs of the measurement are discarded.
The following theorem generalizes Theorem~\ref{thm:nsumboxconstruction} to non-maximal-stabilizer-based sum boxes.
The proof is omitted, as it is the same as the one provided for Theorem~\ref{thm:nsumboxconstruction}.

\begin{theorem}
    \label{thm:kn_sum_box_construction}
    Let $\bG \in \Fq^{2N \times \kappa}$ and $\bG^\perp \in \Fq^{2N \times 2N - \kappa}$ be matrices satisfying the conditions of Proposition~\ref{prop:measurementoutcome}, \ie such that 
    \begin{itemize}
        \item $\bG = \bG^\perp \ppsmatrix{\bI_\kappa \\ \bzero}$,
    
        \item $\bG^\top \bJ \bG^\perp = \bzero$,

        \item there exists $\bH \in \Fq^{2N \times \kappa}$ such that $\ppmatrix{\bG^\perp & \bH}$ is full rank.
    \end{itemize}
    Then there exists a stabilizer-based construction for a $(\kappa,N)$-sum box over $\Fq$ with transfer matrix
    \[
    \bM = \ppmatrix{\bzero_{\kappa \times 2N - \kappa} & \bI_{\kappa}} \ppmatrix{\bG^\perp & \bH}^{-1} \in \Fq^{\kappa \times 2N}.
    \]
\end{theorem}

The following example shows the implementation of a $(1,2)$-sum box over $2$ qubits.

\begin{example}
    \label{ex:kn_sum_box}
    Suppose we have two parties, Tx1 and Tx2, both possessing two bits $\ba = (x_1, x_3),\ \bb = (x_2, x_4)\in \F_2^2$, respectively.
    Let $\cS = \langle \bW(1,1,0,0) \rangle$ be a 1-dimensional stabilizer over 2 qubits.
    Notice that its generator is the restriction to the first two qubits of $\bW(1,1,0,0,0,0)$, \ie the first Weyl operator that fixes $\ket{\varphi'}$ defined in Example~\ref{ex:non_max_stab}.
    The stabilizer has a generator matrix that can be completed to a symplectic matrix $\bF$ as follows:
    \begin{equation*}
    \bF = 
    \ppmatrix{1 & 0 & 0 & 1 \\ 
              1 & 0 & 0 & 0 \\ 
              0 & 1 & 0 & 0 \\ 
              0 & 1 & 1 & 0}.
    \end{equation*}
    Thus, we can choose matrices $\bG^\top$ and $\bH$ as
    \begin{equation*}
    \bG^\top = 
    \ppmatrix{1 & 0 & 1 \\ 
              1 & 0 & 0 \\ 
              0 & 1 & 0 \\ 
              0 & 1 & 0}, \quad 
    \bH = \ppmatrix{0 \\ 0 \\ 0 \\ 1}.
    \end{equation*}
    The output of the box is then given by 
    \[
    \ppmatrix{\bzero_{3 \times 1} & \bI_1} \ppmatrix{\bG^\top & \bH}^{-1} \bx = 
    \ppmatrix{0 & 0 & 0 & 1} 
    \ppmatrix{0 & 1 & 0 & 0 \\ 
              0 & 0 & 1 & 0 \\ 
              1 & 1 & 0 & 0 \\ 
              0 & 0 & 1 & 1}  
    \ppmatrix{x_1 \\ x_2 \\ x_3 \\ x_4} = 
    x_3 + x_4.
    \]
    Now, notice that the Weyl operator $\bW(\ba) \otimes \bW(\bb)$ applied by the transmitters together and the non-maximal stabilizer $\cS$ correspond, respectively, to the Weyl operator $\bW(\bs)$ and the stabilizer $\cS(\cV)$ in Example~\ref{ex:non_max_stab}.
    As the outputs match, it is easy to see that the (1,2)-sum box simplifies the description of the non-maximal-stabilizer construction.
\end{example}

\begin{remark}
    \label{rem:max_to_non_max}
    From the discussion above, one can see that maximal stabilizers exhaust the scope of stabilizer-based constructions for black boxes of the form $\bY = \bM \bX$, where $\bM \in \Fq^{\kappa \times 2N}$.
    More precisely, consider matrices $\bG,\bH$ that satisfy the conditions of Theorem~\ref{thm:nsumboxconstruction}, \ie the matrices associated to a stabilizer-based construction for an $N$-sum box.
    Then there is a correspondence between $\bU_{\bG,\bH}$ and the maximal stabilizer $\cS$, where $\bU_{\bG,\bH}$ is the unitary that prepares the initial state $\ket{\bzero}_\cW = \bU_{\bG,\bH} \ket{\bzero}$ for the $N$-sum box and $\cS = \cS(\colspan{\bG})$.
    If we want $N-\kappa$ of the outputs of an $N$-sum box to be discarded, we can consider the initial state $\ket{\bzero}_{\cW'} = \bU_{\bG,\bH} \paren*{\ket{\bzero} \otimes \frac{\bI_{q^{N-\kappa}}}{q^{N-\kappa}}}$, as operations on the qudits initially prepared in the mixed state do not affect their mixedness and measuring them at the end of the sum-box operations outputs random digits.
    Thus, $(\kappa,N)$-sum boxes can be obtained from $N$-sum boxes by changing the initial (unentangled) state to a mixed state and their description (as in Theorem~\ref{thm:kn_sum_box_construction}) can be formalized by using a submatrix of $\bG$ as the generator for the non-maximal stabilizer.
\end{remark}

\subsection{Applications of $(\kappa,N)$-Sum Boxes}
\label{sec:applications_kn_sum_box}

The utilization of a $(\kappa,N)$-sum box may be perceived as inefficient due to the presence of non-informative qudits among the total $N$ qudits being transferred.
Nevertheless, it has some meaningful applications, \eg it is useful to make symmetric QPIR protocols which use the $N$-sum box construction without the need for shared randomness among the servers like in classical PIR \cite{Chen_Jia_Wang_Jafar_NGCSA}.

\begin{remark}
    Notice that, in this case, we do not need to provide any additional information to the servers to provide symmetry for the protocol.
    In fact, we only need to change the initial entangled state distributed to the servers, rather than providing additional shared randomness to the servers which requires additional computations and storage.
\end{remark}

For instance, consider the output of a QCSA scheme as given in Equation~\eqref{eq:qcsa_output}, \ie
\[
\ppmatrix{\delta^1_1,\cdots,\delta^1_L,\bm{\nu}^1_{(\leftharpoondown)}, \delta^2_1,\cdots,\delta^2_L,\bm{\nu}^2_{(\leftharpoondown)}}^\top,
\]
where $\bm{\nu}^b_{(\leftharpoondown)}$ represents the last $\lfloor N/2\rfloor -L$ symbols of the interference vector $\bm{\nu}^b = \paren*{\nu^b_1,\ldots,\nu^b_{N-L}}$ for $b \in [2]$.
This interference contains linear combinations of other files stored in the server that might be retrieved over many rounds by a malicious user.
Thus, we want to \qmarks{hide} the interference outputs by replacing them with random digits, which can be achieved in the following way.
Some details are omitted, as the procedure is similar to the one shown in Appendix~\ref{app:proofQCSA}.

Consider two GRS codes $\cC_1 = \mathrm{GRS}^q_{N,L}(\a,\bu)$ and $\cC_2 = \mathrm{GRS}^q_{N,L}(\a,\bv)$ where $\bv$ is generated by $\bu$ according to Equation~\eqref{eq:def_v}, then we have that $\cC_2 \subseteq \cC_1^\perp$ or, equivalently, $\bG_{\cC^\perp_1}^\top \bG_{\cC_2} = \bzero_{(N-L) \times L}$ for $L \leq \floor{N/2}$.
Similarly, we have that $\cC_1 \subseteq \cC_2^\perp$ or, equivalently, $\bG_{\cC^\perp_2}^\top \bG_{\cC_1} = \bzero_{N-L \times L}$.
Let
\[
\bG_{\cC_b^\perp} = \ppmatrix{\bG_{\cC_b} & \bG'_{\cC_b}},\ b \in [2].
\]
Consider the matrices
\begin{align*}
    \bG & = \ppmatrix{
    \bG_{\cC_1} & \bzero \\ 
    \bzero & \bG_{\cC_2}
    } \in \Fq^{2N \times 2L}, \\
    \bG^\perp & = \ppmatrix{
    \bG_{\cC_1} & \bzero & \bG'_{\cC_1} & \bzero \\ 
    \bzero & \bG_{\cC_2} & \bzero & \bG'_{\cC_2}
    } \in \Fq^{2N \times 2(N-L)}, \\
    \bH & = \ppmatrix{
    \tilde{\bH}^\bu_\sC & \bzero \\
    \bzero & \tilde{\bH}^\bv_\sC
    } \in \Fq^{2N \times 2L},
\end{align*}
where $\tilde{\bH}^\bu_\sC$ and $\tilde{\bH}^\bv_\sC$ are the Cauchy submatrices of QCSA matrices defined in Equation~\eqref{eq:MCV_Blocks}.
This matrix satisfies the conditions of Theorem~\ref{thm:kn_sum_box_construction}, then there exists a $(2L,N)$-sum box with transfer matrix
\[
\bM = \ppmatrix{\bzero_{\kappa \times 2N - \kappa} & \bI_{\kappa \times \kappa}} \ppmatrix{\bG^\perp & \bH}^{-1}.
\]
Clearly there exists a permutation matrix $\bP_\pi$ as defined in Equation~\eqref{eq:def_perm} such that
\[
\ppmatrix{\bG^\perp & \bH} = 
\ppmatrix{
\bG_{\mathrm{QCSA}^q_{N,L}(\a,\bu,\bff)} & \bzero \\ 
\bzero & \bG_{\mathrm{QCSA}^q_{N,L}(\a,\bv,\bff)}
} \bP_\pi.
\]
The transfer matrix is given by
\[
\bM_{\mathrm{QCSA}} = \ppmatrix{\begin{array}{cc|cc}
\bI_L & \bzero & \bzero & \bzero \\
\bzero &\bzero & \bI_L &\bzero
\end{array}}
\ppmatrix{
\bQ_N^{\bu} & \bzero\\
\bzero & \bQ_N^{\bv}
}^{-1}.
\]
It is easily verified that the output of the $(2L,N)$-sum box is just
\[
\ppmatrix{\delta^1_1,\ldots,\delta^1_L, \delta^2_1,\ldots,\delta^2_L}^\top,
\]
without the interference terms that appear in Equation~\eqref{eq:qcsa_output}.
The QCSA scheme is thus symmetric, as the remaining $N - 2L$ outputs are randomized by construction.

\section{Conclusion}

In this paper we considered a quantum protocol based on the stabilizer formalism to improve the communication rates of classical many-to-one networks employing quantum communication.
Although the protocol proposed in this study may be considered obvious from a quantum coding-theoretic perspective, as noted in Remark~\ref{rem:quantumtrivialprotocol}, it may not be as straightforward for the classical-coding community.
The $N$-sum box represents a quantum black-box with classical inputs and classical outputs that can communicate $N$ sums from $N$ transmitters, each holding 2 symbols from a finite field.
In a classical setting, assuming that the transmitters cannot talk to each other, the only way to achieve the communication of $N$ linear combinations of $2N$ inputs is to send all the $2N$ symbols from the transmitters to the receiver, compared to sending only $N$ qudits with the $N$-sum box.
We presented an application of such black box to improve the rates of CSA-based PIR and SDMM using the underlying superdense-coding gain.
Furthermore, we showed how we can create a symmetric CSA-based QPIR protocol by using a non-maximal-stabilizer-based $(\kappa,N)$-sum box instead of sharing common randomness among the servers.

\clearpage


\appendices

\section{Proof of lemma~\ref{lemm:colswap}}
\label{app:colswaplemma}

\begin{proof}
    The first two claims \eqref{eq:colswapequivalence} and \eqref{eq:colswapsso} are true not just for a particular $\bgS$, but for every diagonal $\bgS$ with elements in $\px*{0,1}$. 
    
    Let $\bgL = \ppsmatrix{\bI - \bgS & \bgS \\ -\bgS & \bI - \bgS}$.
    First, we prove Equation~\eqref{eq:colswapequivalence} by testing the LIT condition on $\bgL$.
    We have that $\det((\bI - \bgS)^2 + \bgS^2) = \det(\bI - 2\bgS + 2\bgS^2) = \det(\bI) = 1$,
    since $\bgS^2 = \bgS$ by the fact that $\bgS$ is a diagonal matrix with entries in $\px*{0,1}$.

    Equation~\eqref{eq:colswapsso} follows by the fact that $\bgL \bJ \bgL^\top = \bJ$ (that is, $\bgL^\top$ is symplectic), which can be easily proved by employing $\bgS^2 = \bgS$.

    Finally, we prove that there exists a signed column-swap operation ($\bgS$) that gives us the desired full-rank left half-matrix of $(\bM')^\top$. 
    For this we proceed according to Algorithm~\ref{alg:colswap}, which tries at most $N$ different signed-swap operations before declaring either success or failure.

    \begin{algorithm}
    \SetKwInOut{Input}{input}\SetKwInOut{Output}{output}
    
    \caption{Signed column swap}
    \label{alg:colswap}
    \Input{$\bM = \ppmatrix{\bM_l & \bM_r} \in \Fq^{N \times 2N}$}
    \Output{$\bM' = \ppmatrix{\bM'_l & \bM'_r} \in \Fq^{N \times 2N} : \det(\bM'_l) \neq 0$}
    $\bM' \gets \bM$\;
    $\bM'_l \gets \bM_l$\;
    $\bM'_r \gets \bM_r$\;
    $i \gets 1$\;
    \While{$i \leq N$}
    {
        \uIf{$(\bM'_l)_{\cdot,i}$ is linearly independent of the first $i-1$ columns of $\bM'_l$}
        {
        $i \gets i+1$\;
        }
        \uElseIf{$(\bM'_r)_{\cdot,i}$ is linearly independent of the first $i-1$ columns of $\bM'_l$ that are already fixed}
        {
        $(\bM'_l)_{\cdot,i},(\bM'_r)_{\cdot,i} \leftarrow (-\bM'_r)_{\cdot,i},(\bM'_l)_{\cdot,i}$\;
        $i \gets i+1$\;
        }
        \Else
        {
        \Return Failure\;
        }
    }
    \Return Success\;
    \end{algorithm}

    If the algorithm exits with success, then we obtain a full-rank $\bM'_l$ as desired. 
    To show that the algorithm cannot fail, let us show that failure would lead to a contradiction. 
    Suppose the algorithm fails and exits with the value $i < N$. At this point, the first $i-1$ columns of $\bM'_l$ are linearly independent, but the $i^{th}$ column of $\bM_l$ and the $i^{th}$ column of $\bM_r$ are each linearly dependent on the first $i$ columns of $\bM_l$.
    Note that since the only manipulations performed by the algorithm are signed-swap operations, by Equation~\eqref{eq:colswapsso} $\bM' \in \cM_o$. 
    The remainder of the proof of (\ref{eq:colswapdet}) uses the following two facts.
    \begin{enumerate}
        \item Since $\bM' \in \cM_o$, the $2N \times N$ matrix $(\bM')^\perp = \bJ \bM'$ spans the null-space of $\bM'$.
        Moreover, $(\bM')^\perp  \in \cM_o$, since multiplication by an invertible matrix preserves rank and
        \[
        (\bJ^\top \bM')^\top \bJ (\bJ^\top \bM') = (\bM')^\top \bJ \bJ \bJ^\top \bM' = (\bM')^\top \bJ \bM' = \bzero.
        \]
        Thus, any matrix whose columns are null-vectors of $\bM$ must be self-orthogonal.

        \item If $\bV = \ppmatrix{\bV_l & \bV_r}^\top \in \cM_o$, then the dot product between the $i^{th}$ row of $\bV_l$ and the $j^{th}$ row of $\bV_r$ is equal to the dot product between the $j^{th}$ row of $\bV_l$ and the $i^{th}$ row of $\bV_r$.
        In fact, $\bV^\top \bJ \bV = \bzero$ implies that
        \[ \begin{split}
            \bV^\top \bJ \bV = \bzero & \implies \ppmatrix{\bV_r & -\bV_l} \ppmatrix{\bV_l & \bV_r}^\top = \bzero \\
            & \implies \bV_r \bV_l^\top - \bV_l \bV_r^\top = \bzero \\
            & \implies (\bV_r \bV_l^\top)_i,j = (\bV_l \bV_r^\top)_i,j,\quad \forall i,j \in [N] \\
            & \implies (\bV_r)_i \cdot (\bV_l)_j = (\bV_l)_i \cdot (\bV_r)_j,\quad \forall i,j \in [N].
        \end{split} \]
    \end{enumerate}
    Since the $(i+1)^{th}$ column of $\bM_l$ and the $(i+1)^{th}$ column of $\bM_r$ are each linearly dependent on the first $i$ columns of $\bM_l$, there exists a $2 \times N$ matrix $\bV$ such that
    \[ \begin{split}
        \bV^\top & = \ppmatrix{\bV_l & \bV_r} \\
        & =
        \begin{pmatrix}[ccccc|ccccc]
            \a_1 & \cdots & \a_i & \a_{i+1} & \bzero & 0 & \cdots & 0 & 0 & \bzero \\
            \b_1 & \cdots & \b_i & 0 & \bzero & 0 & \cdots & 0 & \b_{i+1} & \bzero
        \end{pmatrix},
    \end{split} \]
    $\a_{i+1} \b_{i+1} \neq 0$, and $(\bM')^\top \bV = \bzero$.
    In fact, since the first $(i+1)$ columns of $\bM_l$ are linearly dependent, there exists a non-trivial linear combination of them with coefficients $\a_1,\ldots,\a_{i+1}$ that produces the zero vector.
    Notice that $\a_{i+1}$ cannot be zero because the first $i$ columns are linearly independent by assumption.
    Thus, the first row of $\bV^\top$ is in the null-space of $\bM'$ and $\a_{i+1} \neq 0$.
    The second row of $\bV^\top$ is similarly in the null-space of $\bM'$ and $\b_{i+1} \neq 0$.
    So, $\bV \in \cM_o$ by fact 1), but the dot product of the first row of $\bV_l$ with the second row of $\bV_r$ is $\a_{i+1} \b_{i+1} \neq 0$, whereas the dot product of the second row of $\bV_l$ with the first row of $\bV_r$ is 0, which contradicts fact 2).
    This contradiction proves Equation~\eqref{eq:colswapdet}.
\end{proof}

\section{Proof of Theorem~\ref{thm:QCSA}}
\label{app:proofQCSA}

For a permutation $\pi : [2N] \rightarrow [2N]$, define,
\begin{equation}
    \label{eq:def_perm}
    \bP_{\pi} \coloneqq \ppmatrix{\be_{\pi(1)} & \be_{\pi(2)} & \cdots & \be_{\pi(2N)}},
\end{equation}
where $\be_{i}$ is the $i^{th}$ column of $\bI_{2N}$. 
When $\bP_{\pi}$ is multiplied to the right of a $2N \times 2N$ matrix $\bA$, the columns of $\bA$ are permuted according to the permutation $\pi$. 
Also, note that $\bP_{\pi}^{-1} = \bP_{\pi^{-1}} = \bP_{\pi}^\top$, 
where $\pi^{-1}$ is the inverse permutation of $\pi$.

Now let us prove Theorem~\ref{thm:QCSA} by specifying the $N$-sum box construction according to Theorem~\ref{thm:nsumboxconstruction}. Define,
\begin{equation}
    \tilde{\bG} = 
    \ppmatrix{
    \bG_{\mathrm{GRS}^q_{N,\lceil N/2\rceil}(\bm{\alpha},\bu)} & \bzero \\ 
    \bzero & \bG_{\mathrm{GRS}^q_{N,\lceil N/2\rceil}(\bm{\alpha},\bv)}
    },
\end{equation}
where $\bG_{\mathrm{GRS}^q_{N,\lceil N/2\rceil}(\bm{\alpha},\bu)}$ and $\bG_{\mathrm{GRS}^q_{N,\lceil N/2\rceil}(\bm{\alpha},\bv)}$ are sub-matrices of $\bQ^\bu_N$ and $\bQ^\bv_N$, respectively.
If $N$ is even, then $\tilde{\bG}$ is a $2N \times N$ square matrix, and we choose $\bG = \tilde{\bG}$ for the $N$-sum box construction. 
If $N$ is odd, then $\tilde{\bG}$ is a $2N \times (N+1)$ matrix, \ie it has an extra column. 
In this case we choose $\bG$ as the $2N\times N$ left-sub-matrix of $\tilde{\bG}$, leaving out the $(N+1)^{th}$ column. 
Thus, 
\begin{equation} \label{eq:Gform}
    \bG = \ppmatrix{
    \bG_{\mathrm{GRS}^q_{N,\lceil N/2\rceil}(\bm{\alpha},\bu)} & \bzero \\ 
    \bzero & \bG_{\mathrm{GRS}^q_{N,\lfloor N/2\rfloor}(\bm{\alpha},\bv)}
    }.
\end{equation}
Note that $\bG$ is a sub-matrix of $\diagonal(\bQ_N^\bu, \bQ_N^\bv) = \ppsmatrix{\bQ_N^\bu & \bzero \\ \bzero & \bQ_N^\bv}$. 
Let $\bH$ in Theorem~\ref{thm:nsumboxconstruction} be chosen as the remaining columns of $\diagonal(\bQ_N^\bu, \bQ_N^\bv)$ after eliminating the columns that are present in $\bG$. Thus,
\begin{align}
    \ppmatrix{\begin{array}{c|c}
    \bG & \bH
    \end{array}}
    & = {\scriptsize\ppmatrix{
    \begin{array}{c|c;{2pt/2pt}c;{2pt/2pt}c;{2pt/2pt}c}
    \tilde{\bG} & 
    \begin{matrix} 
    \tilde{\bH}^\bu_\sC & \tilde{\bH}^\bu_\sV \\ \bzero  & \bzero
    \end{matrix} & 
    \begin{matrix} 
    \bzero \\ \tilde{\bH}^\bv_\sC
    \end{matrix} & 
    {\color{black!50} \tilde{\bG}_{\cdot,N+1}} & 
    \begin{matrix}
    \bzero \\ \tilde{\bH}^\bv_\sV
    \end{matrix} 
    \end{array}}}
    \label{eq:HGform}\\
    & =
    \ppmatrix{
    \bQ_N^\bu & \bzero\\
    \bzero & \bQ_N^\bv
    }
    \bP_{\pi}.
\end{align}
The column $\tilde{\bG}(:,N+1)$ does not appear in Equation~\eqref{eq:HGform} (\ie it is  empty) if $N$ is even, because in that case $\tilde{\bG}$ has only $N$ columns.
The matrices $\tilde{\bH}^\bu_\sC, \tilde{\bH}^\bv_\sC$ and $\tilde{\bH}^\bu_\sV, \tilde{\bH}^\bv_\sV$ are the respective Cauchy and Vandermonde matrices specified in Equation~\eqref{eq:MCV_Blocks}. 
Evidently, the $2N\times 2N$ matrix $\ppmatrix{\begin{array}{c|c} \bG & \bH \end{array}}$ has full rank $2N$ as required for an $N$-sum box construction Theorem~\ref{thm:nsumboxconstruction}, since the rank of a block-diagonal matrix is the sum of the ranks of its blocks and $\bQ_N^\bu, \bQ_N^\bv$ have full rank $N$, while $\bP_\pi$ is an invertible square matrix which preserves the same rank. 
The permutation $\pi$ is explicitly expressed as
\begin{equation} \label{eq:HGperm}
    \begin{split}
        \pi = & (\pi(1),\ldots,\pi(2N)) \\
        = & \big(L+1, L+2, \ldots, L+\lceil N/2\rceil, \\
        & N+L+1, N+L+2, \ldots,N+L+\lfloor N/2\rfloor\\
        & 1,2,\ldots, L, L+\lceil N/2\rceil+1, L+\lceil N/2\rceil+2, \ldots, N,\\
        & N+1, N+2, \ldots, N+L,\\
        & N+L+\lfloor N/2\rfloor +1,N+L+\lfloor N/2\rfloor +2, \ldots, 2N\big).
    \end{split}
\end{equation}
Note that $\bP_{\pi}$ moves the columns of $\diagonal(\bQ_N^\bu, \bQ_N^\bv)$ with indices in $\bparen*{L+1: L+\lceil {N}/{2} \rceil} \cup \bparen*{N+L+1: N+L+\lfloor{N}/{2}\rfloor}$ to the left-most part.

Now, notice that $\bG$ is SSO, \ie $\bG \in \cM_o$, since $\ppmatrix{\begin{array}{c|c} \bG & \bH \end{array}}$ has full rank $2N$ by the previous discussion and $\bG^\top \bJ \bG = \bzero$ by Equation~\eqref{eq:dualsub}.
It follows that this $N$-sum box is feasible according to Theorem~\ref{thm:nsumboxconstruction}.

For the specified $\bG, \bH$ matrices, the transfer matrix $\bM$ of the resulting $N$-sum box is
\begin{equation}\label{eq:verify}
    \bM = \ppmatrix{\bzero_N & \bI_N} = 
    \ppmatrix{\bG & \bH} =
    \ppmatrix{\bzero_N & \bI_N} \bP_\pi^{-1} \ppmatrix{\bQ_N^\bu & \bzero \\ \bzero & \bQ_N^\bv}^{-1}
\end{equation}
Recall that $\bP_{\pi}^{-1} = \bP_{\pi^{-1}}$. 
The expression of the permutation $\pi^{-1}$ can be written explicitly as
\begin{align}
    \pi^{-1} = & \big(N+1, N+2, \cdots, N+L, 1, 2, \cdots, \lceil N/2 \rceil,\notag\\
    &N+L+1, N+L+2, \cdots, N + \lfloor N/2 \rfloor,\notag\\
    &N + \lfloor N/2 \rfloor + 1, N + \lfloor N/2 \rfloor + 2, \cdots, N + \lfloor N/2 \rfloor + L, \notag\\
    &\lceil N/2 \rceil + 1, \lceil N/2 \rceil + 2, \cdots, N,\\
    &N + L + \lfloor N/2 \rfloor + 1, N + L + \lfloor N/2 \rfloor + 2, \cdots, 2N\big).\notag
\end{align}
Now it is easily verified that $\ppmatrix{\bzero_N & \mathbf{I}_{N}}\bP_{\pi}^{-1}$ in Equation~\eqref{eq:verify} is  the same as the  matrix in Equation~\eqref{eq:thmQCSA}. 
Theorem~\ref{thm:QCSA} is thus proved.$\hfill\square$

\section{Proof of Corollary~\ref{cor:mdsxstpir}}
\label{app:proof_mdsxstpir}

Before we move forward to the general scheme translation, let us start from a specific example, $(N=5,M,K=2,X=1,T=1)$-MDSXSTPIR to give an intuition of how are the messages stored among servers and what do the parameters mean. 
Other details are omitted and can be found in \cite[Section IV]{Jia_Jafar_MDSXSTPIR}.

\subsubsection{$(N=5,M,K=2,X=1,T=1)$ MDSXSTPIR}
In this example, we have $L = N-(X+T+K -1) = 2$. Thus, in the classical setting, we will use the CSA scheme based on the $\mathrm{CSA}_{N=5,L=2}^{q}(\bm{\alpha}, \bff)$ matrix. 
In this setting, for the $b^{th}$ instance of a CSA scheme, there are $M$ messages $\bW^{b,1}, \ldots, \bW^{b,M}$. For any $b \in [2],\ i \in [M]$, message 
\begin{equation}
    \bW^{b,i} = \ppmatrix{ W^{b,i}_{1,1} & W^{b,i}_{1,2} \\W^{b,i}_{2,1} & W^{b,i}_{2,2} } \in \Fq^{L\times K}
\end{equation}
consists of $LK=4$ symbols from $\Fq$.

For any $l \in [L],\ k \in [K],\ b \in [2]$, let 
\begin{align}
    \bW^b_{l,k} = \ppmatrix{W^{b,1}_{l,k} &  W^{2}_{l,k} & \cdots  & W^{b,M}_{l,k}} \in \Fq^{1 \times M},
\end{align}
be the $1 \times M$ row vector that contains the $(l,k)^{th}$ symbol of all the $M$ messages in the $b^{th}$ instance of a CSA scheme.

There are $N=5$ servers and the storage at server $n \in [N]$ is $S_n = (S^1_{n,1}, S^1_{n,2}, S^2_{n,1}, S^2_{n,2})$, where 
\begin{equation}
    S^b_{n,l} = \frac{1}{(f_1 - \alpha_n)^2} \bW^b_{l,1} + \frac{1}{f_1 - \alpha_n}\bW^b_{l,2} + \bZ^b_{l} \in \Fq^{1\times K},\ l \in [2],
\end{equation}
and $\bZ^b_1, \bZ^b_2$ are independent random noise vectors that are uniform over $\Fq^{1\times K}$. 
Thus, it is obvious that the storage cost at each server is $2LM = 4M$ $q$-ary symbols. 
Comparing with the replicated storage, the storage cost is reduced to $\frac{1}{K} = \frac{1}{2}$, since the $M$ messages consist of $2LKM=8M$ $q$-ary symbols in total. 
Also, due to the fact that the messages are padded with random noise, any $X=1$ server learns nothing about the $M$ messages.

A user wishes to retrieve the $\theta^{th}$ message by querying the $N=5$ servers, without letting any $T=1$ server learn anything about $\theta$. 
In the scheme of \cite{Jia_Jafar_MDSXSTPIR}, for one instance of MDSXSTPIR there are $K=2$ rounds of retrieval. 
In each round, the user downloads 1 $q$-ary answer symbol from each server. 
In the first round the user recovers $W_{1,1}^{\theta}, W_{2,1}^{\theta}$, while in the second round the user recovers $W_{1,2}^{\theta}, W_{2,2}^{\theta}$. 
Thus, in each round, the user recovers $L = 2$ symbols of the desired message $\bW^{\theta}$ by downloading $5$ symbols in total from the servers.
The total downloaded symbols are $NK = 10$, while the retrieved desired symbols are $LK = 4$, thus achieving a rate of 
\[
R_\sC = \frac{2}{5} = \frac{LK}{NK} = \frac{N - (X+T+K-1)}{N}.
\]
Clearly, the rate is the same for two instances of the same scheme.

\subsubsection{QCSA scheme for general MDSXSTPIR}
Let us briefly introduce the general CSA code-based classical MDSXSTPIR scheme \cite{Jia_Jafar_MDSXSTPIR}, whereupon the translation to the quantum scheme is obtained as described in Section~\ref{subsec:CSA2QCSA}. 
In the classical scheme there are $K$ rounds of retrieval. 
In round $\kappa, \kappa \in [K]$ a $\mathrm{CSA}_{N,L}$ scheme is applied, where $0 < L = N - (X + T + K -1)$. 
Specifically, according to \cite[Equations~(65),(66)]{Jia_Jafar_MDSXSTPIR}, the $N$ answer symbols from the $N$ servers in the $\kappa^{th}$ round of the $b^{th}$ instance of a CSA scheme are
\begin{align}
    \underbrace{\ppmatrix{
    A_1^{b,(\kappa)}\\
    \vdots\\
    A_N^{b,(\kappa)}
    }}_{\bA^{b,(\kappa)})}
    &=
    \underbrace{\ppmatrix{\begin{array}{ccc|cccc}
    \frac{1}{f_1-\alpha_1}&\cdots&\frac{1}{f_L-\alpha_1}&1&\alpha_1&\cdots&\alpha_1^{N-L-1}\\
    \frac{1}{f_1-\alpha_2}&\cdots&\frac{1}{f_L-\alpha_2}&1&\alpha_2&\cdots&\alpha_2^{N-L-1}\\
    \vdots&\vdots&\vdots&\vdots&\vdots&\vdots&\vdots\\
    \frac{1}{f_1-\alpha_N}&\cdots&\frac{1}{f_L-\alpha_N}&1&\alpha_N&\cdots&\alpha_N^{N-L-1}\\
    \end{array}}}_{\bG_{\mathrm{CSA}_{N,L}(\bm{\alpha},\bff)}}
    \ppmatrix{
    W_{1,\kappa}^{b,\theta}\\
    \vdots\\
    W_{L,\kappa}^{b,\theta}\\
    *\\
    \vdots\\
    *
    }
    +
    \underbrace{
    \ppmatrix{
    \sum_{l\in[L]}\sum_{k \in [\kappa - 1]}\frac{1}{(f_l - \alpha_1)^{\kappa - k + 1}}W_{l,k}^{b,\theta}\\
    \sum_{l\in[L]}\sum_{k \in [\kappa - 1]}\frac{1}{(f_l - \alpha_2)^{\kappa - k + 1}}W_{l,k}^{b,\theta}\\
    \vdots\\
    \sum_{l\in[L]}\sum_{k \in [\kappa - 1]}\frac{1}{(f_l - \alpha_N)^{\kappa - k + 1}}W_{l,k}^{b,\theta}
    }}_{\bR^{b,(\kappa)}}\\
    &=
    \underbrace{\ppmatrix{\begin{array}{ccc|cccc}
    \frac{1}{f_1-\alpha_1}&\cdots&\frac{1}{f_L-\alpha_1}&1&\alpha_1&\cdots&\alpha_1^{N-L-1}\\
    \frac{1}{f_1-\alpha_2}&\cdots&\frac{1}{f_L-\alpha_2}&1&\alpha_2&\cdots&\alpha_2^{N-L-1}\\
    \vdots&\vdots&\vdots&\vdots&\vdots&\vdots&\vdots\\
    \frac{1}{f_1-\alpha_N}&\cdots&\frac{1}{f_L-\alpha_N}&1&\alpha_N&\cdots&\alpha_N^{N-L-1}\\
    \end{array}}}_{\bG_{\mathrm{CSA}_{N,L}(\bm{\alpha},\bff)}}
    \underbrace{\ppmatrix{
    W_{1\kappa}^{b,\theta} + F_1^{b,(\kappa)} &= \delta_1^{b,(\kappa)}\\
    \vdots\\
    W_{L\kappa}^{b,\theta} + F_L^{b,(\kappa)} &= \delta_L^{b,(\kappa)}\\
    * + F_{L+1}^{b,(\kappa)} &= \nu_1^{b,(\kappa)}\\
    \vdots\\
    * + F_N^{b,(\kappa)} &= \nu_{N-L}^{b,(\kappa)}
    }}_{\bX_{\delta,\nu}^{b,(\kappa)}(i)},\label{eq:mdsxstpir_qform}\\
    &=\bG_{\mathrm{CSA}_{N,L}(\bm{\alpha},\bff)} \bX_{\delta,\nu}^{b,(\kappa)}
\end{align}
where  $W_{1,\kappa}^{b,\theta}, \cdots, W_{L,\kappa}^{b,\theta}$ are the $L = N - (X+T+K - 1)$ symbols of the desired message $\bW^{b,\theta}$, $*$ represents undesired information (interference) whose explicit expression is redundant, $\bR^{b,(\kappa)}$ is an $N \times 1$ column vector which only depends on the message symbols that are decodable from the previous $\kappa - 1$ rounds (to be proved for the quantum scheme later), and $\bF^{b,(\kappa)} = \ppmatrix{F_1^{b,(\kappa)} & F_2^{b,(\kappa)} & \cdots & F_{N}^{b,(\kappa)}}^{\top}$ is acquired by projecting the $\bR^{b,(\kappa)}$ along the columns of $\bG_{\mathrm{CSA}_{N,L}(\bm{\alpha},\bff)}$, \ie
\begin{align}
    \bF^{b,(\kappa)} = \bG_{\mathrm{CSA}_{N,L}(\bm{\alpha},\bff)}^{-1} \bR^{b,(\kappa)}.
\end{align}

With the answers in the form of Equation~\eqref{eq:mdsxstpir_qform} it is clear that in round $\kappa \in [K]$, with $2$ instances of the CSA scheme, $2L$ ``desired" symbols $\delta_1^{1,(\kappa)}, \ldots, \delta_L^{1,(\kappa)}, \delta_1^{2,(\kappa)}, \ldots, \delta_L^{2,(\kappa)}$ can be recovered by using the $N$-sum Box specified in Theorem~\ref{thm:QCSA} when $L \leq \frac{N}{2}$, according to Section~\ref{subsec:CSA2QCSA}. 

Note that when $\kappa = 1$ we have that $\bR^{b,(\kappa)} = \bzero$, which implies $\bF^{b,(\kappa)} = \bzero$. 
Thus, in the first round, $2L$ message symbols $W_{1,1}^{1,\theta} = \delta_1^{1,(1)}, \ldots, W_{L,1}^{1,\theta} = \delta_L^{1,(1)}, W_{1,1}^{2,\theta} = \delta_1^{2,(1)}, \ldots, W_{L,1}^{2,\theta} = \delta_L^{2,(1)}$ are decodable.

Let us then prove the decodability of $2L$ symbols $W_{1,\kappa}^{1,\theta}, \ldots, W_{L,\kappa}^{1,\theta}, W_{1,\kappa}^{2,\theta}, \ldots, W_{L,\kappa}^{2,\theta}$ in the $\kappa^{th}$ round by induction for $\kappa \in \px{2,\ldots,K}$. 

Suppose that in the first $(\kappa - 1)$ rounds $2L(\kappa - 1)$ symbols $\px*{W_{l,k}^{b,\theta}}_{l \in [L], k \in [\kappa - 1], b \in [2]}$ are successfully decoded, then the user can find the values $\bR^{b,(\kappa)}, b\in [2]$ in the $\kappa^{th}$ round as they only depend on the message symbols that are decoded in the previous $(\kappa - 1)$ rounds. 
The user is then able to find $\bF^{b,(\kappa)}, b \in [2]$, as $\bF^{b,(\kappa)}$ is just a deterministic function of $\bR^{b,(\kappa)}$. 
Then by subtracting $F_l^{b,(\kappa)}$ from $\delta_l^{b,(\kappa)}$, the user is able to recover $W_{l,\kappa}^{b,\theta}$ for all $l \in [L], b \in [2]$. 
That is to say, in the first $\kappa$ rounds, the user is able to decode $2L\kappa$ symbols $\px*{W_{l,k}^{b,\theta}}_{l \in [L], k \in [\kappa], b \in [2]}$. 
The induction is thus completed. 

Thus, the rate when $L = N - (X+T+K - 1) \leq \frac{N}{2}$ can be computed as $\frac{2L}{N} = 2\paren*{1 - \frac{X+T+K - 1}{N}}$. 

For the case where $L > \frac{N}{2}$, we can eliminate ``redundant'' servers and construct a rate 1 scheme as mentioned in Section~\ref{subsec:CSA2QCSA}. 
Corollary~\ref{cor:mdsxstpir} is thus proved. $\hfill\square$

\subsection{Proof of Corollary~\ref{cor:sdmm}}\label{proof:sdmm}

Without loss of generality, let us prove that the specified rate is achievable for retrieving the $(1,1)^{th}$ entry of the $L$ products $\bA_1\bB_1, \ldots, \bA_L\bB_L$. 
Other entries are similarly retrieved.

For any $l \in [L]$, let $\bC_l = \bA_l\bB_l$, and let $c_l$ be the $(1,1)^{th}$ entry of the matrix $\bC_l$. 
According to \cite[Equations~(103),(110),(112)]{Jia_Jafar_SDMM}, after applying a $\mathrm{CSA}_{N,L}$ scheme with $0 < L = N - (X_A + X_B)$, the $N$ answer symbols from the $N$ servers regarding the $(1,1)^{th}$ entry of the $L$ products are
\begin{align}
    \underbrace{\ppmatrix{
    A_1^{b}\\
    \vdots\\
    A_N^{b}
    }}_{\bA^b}
    &=
    \underbrace{\ppmatrix{\begin{array}{ccc|cccc}
    \frac{1}{f_1-\alpha_1}&\cdots&\frac{1}{f_L-\alpha_1}&1&\alpha_1&\cdots&\alpha_1^{N-L-1}\\
    \frac{1}{f_1-\alpha_2}&\cdots&\frac{1}{f_L-\alpha_2}&1&\alpha_2&\cdots&\alpha_2^{N-L-1}\\
    \vdots&\vdots&\vdots&\vdots&\vdots&\vdots&\vdots\\
    \frac{1}{f_1-\alpha_N}&\cdots&\frac{1}{f_L-\alpha_N}&1&\alpha_N&\cdots&\alpha_N^{N-L-1}\\
    \end{array}}}_{\bG_{\mathrm{CSA}_{N,L}(\bm{\alpha},\bff)}}
    \underbrace{\ppmatrix{
    c^b_1 &= \delta^b_1\\
    \vdots\\
    c^b_L &= \delta^b_L\\
    * &= \nu^b_1\\
    \vdots\\
    * &= \nu^b_{N-L}
    }}_{\bX^b_{\delta,\nu}},
\end{align}
where $c^b_1 = \delta^b_1, \cdots, c^b_L = \delta^b_L$ are the $(1,1)^{th}$ entry of the $L$ products in the $b^{th}$ instance of the CSA scheme, \ie the $L$ desired symbols, and $*$ denotes the undesired information (interference) whose explicit expressions are redundant.

Thus, with $2$ instances of the CSA scheme, $2L$ desired symbols can be recovered by using the $N$-sum Box specified in Theorem~\ref{cor:sdmm} when $L = N - (X_A + X_B) \leq \frac{N}{2}$, according to Section~\ref{subsec:CSA2QCSA}. 
The rate of retrieving one entry of the $L$ products is thus $\frac{2L}{N} = 2\bigg(1 - \frac{X_A + X_B}{N}\bigg)$ when $L \leq \frac{N}{2}$. 
Again, we can achieve rate $1$ by eliminating ``redundant'' servers when $L > \frac{N}{2}$, as mentioned in Section~\ref{subsec:CSA2QCSA}. 
Corollary~\ref{cor:sdmm} is thus proved. $\hfill\square$

\clearpage


\bibliographystyle{IEEEtran}
\bibliography{main}

\begin{thebibliography}{10}
\providecommand{\url}[1]{#1}
\csname url@samestyle\endcsname
\providecommand{\newblock}{\relax}
\providecommand{\bibinfo}[2]{#2}
\providecommand{\BIBentrySTDinterwordspacing}{\spaceskip=0pt\relax}
\providecommand{\BIBentryALTinterwordstretchfactor}{4}
\providecommand{\BIBentryALTinterwordspacing}{\spaceskip=\fontdimen2\font plus
\BIBentryALTinterwordstretchfactor\fontdimen3\font minus
  \fontdimen4\font\relax}
\providecommand{\BIBforeignlanguage}[2]{{%
\expandafter\ifx\csname l@#1\endcsname\relax
\typeout{** WARNING: IEEEtran.bst: No hyphenation pattern has been}%
\typeout{** loaded for the language `#1'. Using the pattern for}%
\typeout{** the default language instead.}%
\else
\language=\csname l@#1\endcsname
\fi
#2}}
\providecommand{\BIBdecl}{\relax}
\BIBdecl

\bibitem{nsumboxarxiv}
M.~Allaix, Y.~Lu, Y.~Yao, T.~Pllaha, C.~Hollanti, and S.~Jafar, ``{$ N $}-sum
  box: An abstraction for linear computation over many-to-one quantum
  networks,'' \emph{arXiv preprint arXiv:2304.07561}, 2023.

\bibitem{QCSA23}
Y.~Lu and S.~A. Jafar, ``Quantum cross subspace alignment codes via the
  {$N$}-sum box abstraction,'' \emph{arXiv preprint arXiv:2304.14676}, 2023.

\bibitem{nazer2007computation}
B.~Nazer and M.~Gastpar, ``Computation over multiple-access channels,''
  \emph{IEEE Transactions on information theory}, vol.~53, no.~10, pp.
  3498--3516, 2007.

\bibitem{dean2008mapreduce}
J.~Dean and S.~Ghemawat, ``{M}ap{R}educe: simplified data processing on large
  clusters,'' \emph{Communications of the ACM}, vol.~51, no.~1, pp. 107--113,
  2008.

\bibitem{Jia_Jafar_SDMM}
Z.~Jia and S.~A. Jafar, ``On the capacity of secure distributed batch matrix
  multiplication,'' \emph{IEEE Transactions on Information Theory}, vol.~67,
  no.~11, pp. 7420--7437, 2021.

\bibitem{ahlswede2000network}
R.~Ahlswede, N.~Cai, S.-Y. Li, and R.~W. Yeung, ``Network information flow,''
  \emph{IEEE Transactions on information theory}, vol.~46, no.~4, pp.
  1204--1216, 2000.

\bibitem{belzner2009network}
M.~Belzner and H.~Haunstein, ``Network coding in passive optical networks,'' in
  \emph{2009 35th European Conference on Optical Communication}.\hskip 1em plus
  0.5em minus 0.4em\relax IEEE, 2009, pp. 1--2.

\bibitem{hayashi2007quantum}
M.~Hayashi, K.~Iwama, H.~Nishimura, R.~Raymond, and S.~Yamashita, ``Quantum
  network coding,'' in \emph{STACS 2007: 24th Annual Symposium on Theoretical
  Aspects of Computer Science, Aachen, Germany, February 22-24, 2007.
  Proceedings 24}.\hskip 1em plus 0.5em minus 0.4em\relax Springer, 2007, pp.
  610--621.

\bibitem{cacciapuotiteleportation2020}
A.~S. Cacciapuoti, M.~Caleffi, R.~Van~Meter, and L.~Hanzo, ``When entanglement
  meets classical communications: Quantum teleportation for the quantum
  internet,'' \emph{IEEE Transactions on Communications}, vol.~68, no.~6, pp.
  3808--3833, 2020.

\bibitem{cacciapuotiquantuinternet2020}
A.~S. Cacciapuoti, M.~Caleffi, F.~Tafuri, F.~S. Cataliotti, S.~Gherardini, and
  G.~Bianchi, ``Quantum internet: Networking challenges in distributed quantum
  computing,'' \emph{IEEE Network}, vol.~34, no.~1, pp. 137--143, 2020.

\bibitem{caleffi2022distributed}
M.~Caleffi, M.~Amoretti, D.~Ferrari, D.~Cuomo, J.~Illiano, A.~Manzalini, and
  A.~S. Cacciapuoti, ``Distributed quantum computing: a survey,'' \emph{arXiv
  preprint arXiv:2212.10609}, 2022.

\bibitem{bennett1992communication}
C.~H. Bennett and S.~J. Wiesner, ``Communication via one-and two-particle
  operators on {Einstein-Podolsky-Rosen} states,'' \emph{Physical review
  letters}, vol.~69, no.~20, p. 2881, 1992.

\bibitem{werner2001all}
R.~F. Werner, ``All teleportation and dense coding schemes,'' \emph{Journal of
  Physics A: Mathematical and General}, vol.~34, no.~35, p. 7081, 2001.

\bibitem{liu2002general}
X.~S. Liu, G.~L. Long, D.~M. Tong, and F.~Li, ``General scheme for superdense
  coding between multiparties,'' \emph{Physical Review A}, vol.~65, no.~2, p.
  022304, 2002.

\bibitem{gorbachev2002teleportation}
V.~N. Gorbachev, A.~I. Trubilko, A.~A. Rodichkina, and A.~I. Zhiliba,
  ``Teleportation and dense coding via a multiparticle quantum channel of the
  {GHZ}-class.'' \emph{Quantum Inf. Comput.}, vol.~2, no.~5, pp. 367--378,
  2002.

\bibitem{el2011network}
A.~El~Gamal and Y.-H. Kim, \emph{Network information theory}.\hskip 1em plus
  0.5em minus 0.4em\relax Cambridge university press, 2011.

\bibitem{tse2005fundamentals}
D.~Tse and P.~Viswanath, \emph{Fundamentals of wireless communication}.\hskip
  1em plus 0.5em minus 0.4em\relax Cambridge university press, 2005.

\bibitem{goldsmith2005wireless}
A.~Goldsmith, \emph{Wireless communications}.\hskip 1em plus 0.5em minus
  0.4em\relax Cambridge university press, 2005.

\bibitem{song2019capacity}
S.~Song and M.~Hayashi, ``Capacity of quantum private information retrieval
  with multiple servers,'' \emph{IEEE Transactions on Information Theory},
  vol.~67, no.~1, pp. 452--463, 2020.

\bibitem{song2019capacitycollusion}
------, ``Capacity of quantum private information retrieval with collusion of
  all but one of servers,'' in \emph{2019 IEEE Information Theory Workshop
  (ITW)}.\hskip 1em plus 0.5em minus 0.4em\relax IEEE, 2019, pp. 1--5.

\bibitem{QTPIR}
------, ``Capacity of quantum private information retrieval with colluding
  servers,'' \emph{IEEE Transactions on Information Theory}, vol.~67, no.~8,
  pp. 5491--5508, 2021.

\bibitem{allaix2020quantum}
M.~Allaix, L.~Holzbaur, T.~Pllaha, and C.~Hollanti, ``Quantum private
  information retrieval from coded and colluding servers,'' \emph{IEEE Journal
  on Selected Areas in Information Theory}, vol.~1, no.~2, pp. 599--610, 2020.

\bibitem{allaix2021quantum}
------, ``High-rate quantum private information retrieval with weakly self-dual
  star product codes,'' in \emph{2021 IEEE International Symposium on
  Information Theory (ISIT)}, 2021, pp. 1046--1051.

\bibitem{QMDSTPIR}
M.~Allaix, S.~Song, L.~Holzbaur, T.~Pllaha, M.~Hayashi, and C.~Hollanti, ``On
  the capacity of quantum private information retrieval from {MDS}-coded and
  colluding servers,'' \emph{IEEE Journal on Selected Areas in Communications},
  vol.~40, no.~3, pp. 885--898, 2022.

\bibitem{QQPIR21}
S.~Song and M.~Hayashi, ``Quantum private information retrieval for quantum
  messages,'' in \emph{2021 IEEE International Symposium on Information Theory
  (ISIT)}, 2021, pp. 1052--1057.

\bibitem{hayashi2021computation}
M.~Hayashi and {\'A}.~V{\'a}zquez-Castro, ``Computation-aided classical-quantum
  multiple access to boost network communication speeds,'' \emph{Physical
  Review Applied}, vol.~16, no.~5, p. 054021, 2021.

\bibitem{sohail2022computing}
M.~A. Sohail, T.~A. Atif, A.~Padakandla, and S.~S. Pradhan, ``Computing sum of
  sources over a classical-quantum mac,'' \emph{IEEE Transactions on
  Information Theory}, vol.~68, no.~12, pp. 7913--7934, 2022.

\bibitem{sohail2022unified}
M.~A. Sohail, T.~A. Atif, and S.~S. Pradhan, ``Unified approach for computing
  sum of sources over cq-mac,'' in \emph{2022 IEEE International Symposium on
  Information Theory (ISIT)}.\hskip 1em plus 0.5em minus 0.4em\relax IEEE,
  2022, pp. 1868--1873.

\bibitem{hayashi2022unified}
M.~Hayashi and S.~Song, ``Unified approach to secret sharing and symmetric
  private information retrieval with colluding servers in quantum systems,''
  \emph{arXiv preprint arXiv:2205.14622}, 2022.

\bibitem{holevo1973bounds}
A.~S. Holevo, ``Bounds for the quantity of information transmitted by a quantum
  communication channel,'' \emph{Problemy Peredachi Informatsii}, vol.~9,
  no.~3, pp. 3--11, 1973.

\bibitem{Gottesman97}
D.~{G}ottesman, ``Stabilizer codes and quantum error correction,'' 1997, {P}hD
  thesis, California Institute of Technology.

\bibitem{AK01}
A.~Ashikhmin and E.~Knill, ``Nonbinary quantum stabilizer codes,'' \emph{IEEE
  Transactions on Information Theory}, vol.~47, no.~7, pp. 3065--3072, 2001.

\bibitem{Ketkar06}
S.~K. A.~Ketkar, A.~Klappenecker and P.~Sarvepalli, ``Nonbinary stabilizer
  codes over finite fields,'' \emph{IEEE Transactions on Information Theory},
  vol.~52, no.~11, pp. 4892--4914, 2006.

\bibitem{CRSS98}
A.~R. Calderbank, E.~M. Rains, P.~W. Shor, and N.~J.~A. Sloane, ``Quantum error
  correction via codes over {${\rm GF}(4)$},'' \emph{IEEE Trans. Inform.
  Theory}, vol.~44, no.~4, pp. 1369--1387, 1998.

\bibitem{SQMAC23}
Y.~Yao and S.~A. Jafar, ``The capacity of classical summation over a quantum
  {MAC} with arbitrarily replicated inputs,'' 2023.

\bibitem{freij2017private}
R.~Freij-Hollanti, O.~W. Gnilke, C.~Hollanti, and D.~A. Karpuk, ``Private
  information retrieval from coded databases with colluding servers,''
  \emph{SIAM Journal on Applied Algebra and Geometry}, vol.~1, no.~1, pp.
  647--664, 2017.

\bibitem{Jia_Sun_Jafar_XSTPIR}
Z.~{Jia}, H.~{Sun}, and S.~A. {Jafar}, ``Cross subspace alignment and the
  asymptotic capacity of {$X$}-secure {$T$}-private information retrieval,''
  \emph{IEEE Transactions on Information Theory}, vol.~65, no.~9, pp.
  5783--5798, Sep. 2019.

\bibitem{Jia_Jafar_MDSXSTPIR}
Z.~Jia and S.~A. Jafar, ``{$X$}-secure {$T$}-private information retrieval from
  {MDS} coded storage with {Byzantine} and unresponsive servers,'' \emph{IEEE
  Transactions on Information Theory}, vol.~66, no.~12, pp. 7427--7438, 2020.

\bibitem{Chen_Jia_Wang_Jafar_NGCSA}
Z.~Chen, Z.~Jia, Z.~Wang, and S.~A. Jafar, ``{GCSA} codes with noise alignment
  for secure coded multi-party batch matrix multiplication,'' \emph{IEEE
  Journal on Selected Areas in Information Theory}, vol.~2, no.~1, pp.
  306--316, 2021.

\bibitem{MR18}
D.~{Maslov} and M.~{Roetteler}, ``Shorter stabilizer circuits via {B}ruhat
  decomposition and quantum circuit transformations,'' \emph{IEEE Transactions
  on Information Theory}, vol.~64, no.~7, pp. 4729--4738, 2018.

\bibitem{PTC22}
T.~Pllaha, O.~Tirkkonen, and R.~Calderbank, ``Binary subspace chirps,''
  \emph{IEEE Transactions on Information Theory}, vol.~68, no.~12, pp.
  7735--7752, 2022.

\bibitem{PRTC20}
T.~Pllaha, N.~Rengaswamy, O.~Tirkkonen, and R.~Calderbank, ``Un-{W}eyl-ing the
  {C}lifford {H}ierarchy,'' \emph{{Quantum}}, vol.~4, p. 370, Dec. 2020.

\bibitem{Omeara}
O.~T. O'Meara, \emph{Symplectic groups}, ser. Mathematical Surveys.\hskip 1em
  plus 0.5em minus 0.4em\relax American Mathematical Society, Providence, R.I.,
  1978, vol.~16.

\bibitem{Macwilliams}
{F. J. MacWilliams} and {N. J. A. Sloane}, \emph{The Theory of Error-Correcting
  Codes}.\hskip 1em plus 0.5em minus 0.4em\relax Elsevier, 1977, vol.~1.

\bibitem{Holevo_Bounds}
A.~S. Holevo, ``Bounds for the quantity of information transmitted by a quantum
  communication channel,'' \emph{Problemy Peredachi Informatsii}, vol.~9,
  no.~3, pp. 3--11, 1973.

\end{thebibliography}

\end{document}